\documentclass{article}
\usepackage{amsmath,amssymb,amsthm}
\usepackage{graphicx,tikz}
\usepackage[ruled,vlined]{algorithm2e}
\usepackage{graphv1}
\usepackage{placeins}
\usetikzlibrary{shapes}
\newtheorem{theorem}{Theorem}
\newtheorem{lemma}{Lemma}
\newtheorem{definition}{Definition}
\newtheorem{corollary}{Corollary}

\newcounter{claimctr}[lemma]
\newcounter{remctr}
\newenvironment{claim}{\medskip\par\noindent\refstepcounter{claimctr}\textit{Claim \theclaimctr.} }{\smallskip\par}
\newenvironment{remark}{\medskip\par\noindent\refstepcounter{remctr}\textbf{Remark \theremctr.} }{\smallskip\par}

\newtheorem{observation}{Observation}

\newcommand{\com}[1]{}
\newenvironment{myenumerate}
{\vspace{-0.05in}
	\begin{enumerate}
		\itemsep -0.04in

	}
	{\end{enumerate}}
\begin{document}
	\title{On the Kernel and Related Problems in Interval Digraphs}
	\author{Mathew C. Francis\footnote{Indian Statistical Institute, Chennai Centre, Chennai 600 029, India. e-mail: \texttt{\{mathew,dalujacob\}@isichennai.res.in}} \and
	Pavol Hell\footnote{School of Computing Science, Simon Fraser University, Burnaby, B.C., Canada V5A1S6. e-mail: \texttt{pavol@sfu.ca}}\and
	Dalu Jacob\footnotemark[1]}
	\date{}
	\maketitle
	\begin{abstract}
		Given a digraph $G$, a set $X\subseteq V(G)$ is said to be an \emph{absorbing set} (resp. \emph{dominating set}) if every vertex in the graph is either in $X$ or is an in-neighbour (resp. out-neighbour) of a vertex in $X$. A set $S\subseteq V(G)$ is said to be an \emph{independent set} if no two vertices in $S$ are adjacent in $G$. A \emph{kernel} (resp. \emph{solution}) of $G$ is an independent and absorbing (resp. dominating) set in $G$. 
		The problem of deciding if there is a kernel (or solution) in an input digraph is known to be NP-complete. Similarly, the problems of computing a minimum cardinality dominating set or absorbing set or kernel, and the problems of computing a maximum cardinality independent set or kernel, are all known to be NP-hard for general digraphs.
		 We explore the algorithmic complexity of these problems in the well known class of \emph{interval digraphs}. A digraph $G$ is an interval digraph if a pair of intervals $(S_u,T_u)$ can be assigned to each vertex $u$ of $G$ such that $(u,v)\in E(G)$ if and only if $S_u\cap T_v\neq\emptyset$. Many different subclasses of interval digraphs have been defined and studied in the literature by restricting the kinds of pairs of intervals that can be assigned to the vertices. We observe that several of these classes, like interval catch digraphs, interval nest digraphs, adjusted interval digraphs and chronological interval digraphs, are subclasses of the more general class of \emph{reflexive interval digraphs}---which arise when we require that the two intervals assigned to a vertex have to intersect. We see as our main contribution the identification of the class of reflexive interval digraphs as an important class of digraphs. We show that while the problems mentioned above are NP-complete, and even hard to approximate, on interval digraphs (even on some very restricted subclasses of interval digraphs called \emph{point-point digraphs}, where the two intervals assigned to each vertex are required to be degenerate), they are all efficiently solvable, in most of the cases linear-time solvable, in the class of reflexive interval digraphs. 
		 
		
		The results we obtain improve and generalize several existing algorithms and structural results for subclasses of reflexive interval digraphs. In particular, we obtain a vertex ordering characterization of reflexive interval digraphs that implies the existence of an $O(n+m)$ time algorithm 
		for computing a maximum cardinality independent set in a reflexive interval digraph, improving and generalizing the earlier known $O(nm)$ time algorithm for the same problem for the interval nest digraphs. (Here $m$ denotes the number of edges in the digraph not counting the self-loops.) We also show that reflexive interval digraphs are kernel perfect and that a kernel in such digraphs can be computed in linear time. This generalizes and improves an earlier result that interval nest digraphs are kernel perfect and that a kernel can be computed in such digraphs in $O(nm)$ time. The structural characterizations that we show for point-point digraphs, apart from helping us construct the NP-completeness/APX-hardness reductions, imply that these digraphs can be recognized in linear time.
		We also obtain some new results for undirected graphs along the way: (a) We describe an $O(n(n+m))$ time algorithm for computing a minimum cardinality (undirected) independent dominating set in cocomparability graphs, which slightly improves the existing $O(n^3)$ time algorithm for the same problem by Kratsch and Stewart; and (b) We show that the {\sc Red-Blue Dominating Set} problem, which is NP-complete even for planar bipartite graphs, is linear-time solvable on \emph{interval bigraphs}, which is a class of bipartite (undirected) graphs closely related to interval digraphs.
	\end{abstract}
	
	\section{Introduction} \label{sec:intro}
	
	Let $H=(V,E)$ be an undirected graph. A set $S\subseteq V(H)$ is said to be an \emph{independent set} in $H$ if for any two vertices $u,v\in S$, $uv\notin E(H)$. A set $S\subseteq V(H)$ is said to be a \emph{dominating set} in $H$ if for any $v\in V(H)\setminus S$, there exists $u\in S$ such that $uv\in E(H)$. A set $S\subseteq V(H)$ is said to be an \emph{independent dominating set} in $H$ if $S$ is dominating as well as independent. Note that any maximal independent set in $H$ is an independent dominating set in $H$, and therefore every undirected graph contains an independent dominating set, which implies that the problem of deciding whether an input undirected graph contains an independent dominating set is trivial. On the other hand, finding an independent dominating set of \emph{maximum cardinality} is NP-complete for general graphs, since independent dominating sets of maximum cardinality are exactly the independent sets of maximum cardinality in the graph. The problem of finding a minimum cardinality independent dominating set is also NP-complete for general graphs~\cite{garey1979computers} and also in many special graph classes (refer~\cite{liu2015independent} for a survey). We study the directed analogues of these problems,
	which are also well-studied in the literature. 
	
	Let $G=(V,E)$ be a directed graph. A set $S\subseteq V(G)$ is said to be an \emph{independent set} in $G$, if for any two vertices $u,v\in S$, $(u,v),(v,u)\notin E(G)$. A set $S\subseteq V(G)$ is said to be an \emph{absorbing (resp. dominating) set} in $G$, if for any $v\in V(G)\setminus S$, there exists $u\in S$ such that $(v,u)\in E(G)$ (resp. $(u,v)\in E(G)$). As any set of vertices that consists of a single vertex is independent and the whole set $V(G)$ is absorbing as well as dominating, the interesting computational problems that arise here are that of finding a maximum independent set, called {\sc Independent-Set}, and that of finding a minimum absorbing (resp. dominating) set in $G$, called {\sc Absorbing-Set} (resp. {\sc Dominating-Set}). A set $S\subseteq V(G)$ is said to be an \emph{independent dominating (resp. absorbing) set} if $S$ is both independent and dominating (resp. absorbing). Note that unlike undirected graphs, the problem of finding a maximum cardinality independent dominating (resp. absorbing) set is different from the problem of finding a maximum cardinality independent set for directed graphs.
	
	Given a digraph $G$, a collection $\{(S_u,T_u)\}_{u\in V(G)}$ of pairs of intervals is said to be an \emph{interval representation} of $G$ if $(u,v)\in E(G)$ if and only if $S_u\cap T_v\neq \emptyset$. A digraph $G$ that has an interval representation is called an \emph{interval digraph}~\cite{das1989interval}. We consider a loop to be present on a vertex $u$ of an interval digraph if and only if $S_u\cap T_u\neq\emptyset$. An interval digraph is a \emph{reflexive interval digraph} if there is a loop on every vertex. Let $G$ be a digraph. If there exists an interval representation of $G$ such that $T_u\subseteq S_u$ for each vertex $u\in V(G)$ then $G$ is called an \emph{interval nest digraph}~\cite{prisner1994algorithms}. If $G$ has an interval representation in which intervals $S_u$ and $T_u$ for each vertex $u\in V(G)$ are required to have a common left end-point, the interval digraphs that arise are called \emph{adjusted interval digraphs}~\cite{feder2012adjusted}. Note that the class of reflexive interval digraphs is a superclass of both interval nest digraphs and adjusted interval digraphs. Another class of interval digraphs, called \emph{interval-point digraphs} arises when the interval $T_u$ for each vertex $u$ is required to be degenerate (it is a point)~\cite{das1989interval}. Note that interval-point digraphs may not be reflexive. We call a digraph $G$ a \emph{point-point digraph} if there is an interval representation of $G$ in which both $S_u$ and $T_u$ are degenerate intervals for each vertex $u$. Clearly, point-point digraphs form a subclass of interval-point digraphs and they are also not necessarily reflexive.
	\medskip
	
	In this paper, we show that the reflexivity of an interval digraph has a huge impact on the algorithmic complexity of several problems related to domination and independent sets in digraphs. In particular, we show that all the problems we study are efficiently solvable on reflexive interval digraphs, but are NP-complete and/or APX-hard even on point-point digraphs. Along the way we obtain new characterizations of both these graph classes, which reveal some of the properties of these digraphs.
	\medskip
	
	An undirected graph is said to be \emph{weakly chordal (or weakly triangulated)} if it does not contain $C_k$ and $\overline{C_k}$ for $k\geq 5$ as induced subgraphs.
	Prisner~\cite{prisner1994algorithms} proved that the underlying undirected graphs of interval nest digraphs are weakly chordal graphs and notes that this means that any algorithm that solves the maximum independent set problem on weakly chordal graphs can be used to solve the {\sc Independent-Set} problem on interval nest digraphs and their reversals.
	Since the problem of computing a maximum independent set can be solved in $O(nm)$ time in weakly chordal graphs~\cite{hss}, it follows that there is an $O(nm)$-time algorithm for the {\sc Independent-Set} problem in interval nest digraphs and their reversals, even when only the adjacency list of the input graph is given.
	
	An undirected graph is a \emph{comparability} graph if its edges can be oriented in such a way that it becomes a partial order. The complements of comparability graphs are called \emph{cocomparability graphs}.
	\medskip
	
	\noindent\textbf{Our results.} We provide a vertex-ordering characterization for reflexive interval digraphs and two simple characterizations for point-point digraphs including a forbidden structure characterization. Our characterization of point-point digraphs directly yields a linear time recognition algorithm for that class of digraphs (note that Muller's~\cite{muller1997recognizing} recognition algorithm for interval digraphs directly gives a polynomial-time recognition algorithm for reflexive interval digraphs). From our vertex-ordering characterization of reflexive interval digraphs, it follows that the underlying undirected graphs of every reflexive interval digraph is a cocomparability graph. Also a natural question that arises here is whether the underlying graphs of reflexive interval digraphs is the same as the class of cocomparability graphs. We show that this is not the case by demonstrating that the underlying graphs of reflexive interval digraphs cannot contain an induced $K_{3,3}$. Thus, Prisner's result mentioned above can be strengthened to say that the underlying undirected graphs of interval nest digraphs and their reversals are $K_{3,3}$-free weakly chordal cocomparability graphs. Also, as the {\sc Independent Set} problem is linear time solvable on cocomparability graphs~\cite{mcconnell1999modular}, the problem is also linear time solvable on reflexive interval digraphs. This improves and generalizes the $O(nm)$-time algorithm for the same problem on interval nest digraphs. In contrast, we prove that the {\sc Independent Set} problem is APX-hard for point-point digraphs. 

	\medskip
	
	Domination in digraphs is a topic that has been explored less when compared to its undirected counterpart. Even though bounds on the minimum dominating sets in digraphs have been obtained by several authors (see the book~\cite{haynes1998domination} for a survey), not much is known about the computational complexity of finding a minimum cardinality absorbing set (or dominating set) in directed graphs. Even for tournaments, the best known algorithm for {\sc Dominating-Set} does not run in polynomial-time~\cite{megiddo1988finding,reid2004domination}. In~\cite{megiddo1988finding}, the authors give an $n^{O(\log n)}$ time algorithm for the {\sc Dominating-Set} problem in tournaments and they also note that {\sc Sat} can be solved in $2^{O(\sqrt{v})}n^K$ time (where $v$ is the number of variables, $n$ is the length of the formula and $K$ is a constant) if and only if the {\sc Dominating-Set} in a tournament can be solved in polynomial time. Thus, determining the algorithmic complexity of the {\sc Dominating-Set} problem even in special classes of digraphs seems to be much more challenging than the algorithmic question of finding a minimum cardinality dominating set in undirected graphs.
	
	For a bipartite graph having two specified partite sets $A$ and $B$, a set $S\subseteq B$ such that $\bigcup_{u\in B} N(u)=A$ is called an \emph{$A$-dominating set}. Note that the graph does not contain an $A$-dominating set if and only if there are isolated vertices in $A$. The problem of finding an $A$-dominating set of minimum cardinality in a bipartite graph with partite sets $A$ and $B$ is more well-known as the {\sc Red-Blue Dominating Set} problem, which was introduced for the first time in the context of the European railroad network~\cite{weihe1998covering} and plays an important role in the theory of fixed parameter tractable algorithms~\cite{dom2009incompressibility}. This problem is equivalent to the well known {\sc Set Cover} and {\sc Hitting Set} problems~\cite{garey1979computers} and therefore, it is NP-complete for general bipartite graphs. The problem remains NP-complete even for planar bipartite graphs~\cite{alber2002fixed}. The class of \emph{interval bigraphs} are closely related to the class of interval digraphs. These are undirected bipartite graphs with partite sets $A$ and $B$ such that there exists a collection of intervals $\{S_u\}_{u\in V(G)}$ such that $uv\in E(G)$ if and only if $u\in A$, $v\in B$, and $S_u\cap S_v\neq\emptyset$.
	\medskip
	
	\noindent\textbf{Our results:} We observe that the problem of solving {\sc Absorbing-Set} on a reflexive interval digraph $G$ can be reduced to the problem of solving {\sc Red-Blue Dominating Set} on an interval bigraph whose interval representation can be constructed from an interval representation of $G$ in linear time. Further, we show that {\sc Red-Blue Dominating Set} is linear time solvable on interval bigraphs (given an interval representation). Thus the problem {\sc Absorbing-Set} (resp. {\sc Dominating-Set}) is linear-time solvable on reflexive interval digraphs, given an interval representation of the digraph as input. If no interval representation is given, Muller's algorithm~\cite{muller1997recognizing} can be used to construct one in polynomial time, and therefore these problems are polynomial time solvable on reflexive interval digraphs even when no interval representation of the input graph is known. In contrast, we prove that the {\sc Absorbing-Set} and {\sc Dominating-Set} problems remain APX-hard even for point-point digraphs.
	\medskip
	
	An independent absorbing set in a directed graph is more well-known as a \emph{kernel} of the graph, a term introduced by Von Neumann and Morgenstern~\cite{morgenstern1953theory} in the context of game theory. They showed that for digraphs associated with certain combinatorial games, the existence of a kernel implies the existence of a winning strategy. Most of the work related to domination in digraphs has been mainly focused on kernels. We follow the terminology in~\cite{prisner1994algorithms} and call an independent dominating set in a directed graph a \emph{solution} of the graph. It is easy to see that a kernel in a directed graph $G$ is a solution in the directed graph obtained by reversing every arc of $G$ and vice versa. 
	Note that unlike in the case of undirected graphs, a kernel need not always exist in a directed graph. Therefore, besides the computational problems of finding a minimum or maximum sized kernel, called {\sc Min-Kernel} and {\sc Max-Kernel} respectively, the comparatively easier problem of determining whether a given directed graph has a kernel in the first place, called {\sc Kernel}, is itself a non-trivial one. In fact, the {\sc Kernel} problem was shown to be NP-complete in general digraphs by Chv\'atal~\cite{chvatal1973computational}. Later, Fraenkel~\cite{fraenkel1981planar} proved that the {\sc Kernel} problem remains NP-complete even for planar digraphs of degree at most 3 having in- and out-degrees at most 2. It can be easily seen that the {\sc Min-Kernel} and {\sc Max-Kernel} problems are NP-complete for those classes of graphs for which the {\sc Kernel} problem is NP-complete. A digraph is said to be \emph{kernel perfect} if every induced subgraph of it has a kernel. Several sufficient conditions for digraphs to be kernel perfect has been explored~\cite{richardson1953solutions,duchet1987sufficient,morgenstern1953theory}. The {\sc Kernel} problem is trivially solvable in polynomial-time on any kernel perfect family of digraphs. But the algorithmic complexity status of the problem of computing a kernel in a kernel perfect digraph also seems to be unknown~\cite{pass2020perfect}. 	Prisner~\cite{prisner1994algorithms} proved that interval nest digraphs and their reversals are kernel-perfect, and a kernel can be found in these graphs in time $O(n^2)$ if a representation of the graph is given. Note that the {\sc Min-Kernel} problem can be shown to be NP-complete even in some kernel perfect families of digraphs that has a polynomial-time computable kernel (see Remark~\ref{rem:minkernel}).
 
	\medskip

	\noindent\textbf{Our results:} We show that reflexive interval digraphs are kernel perfect and hence the {\sc Kernel} problem is trivial on this class of digraphs. We construct a linear-time algorithm that computes a kernel in a reflexive interval digraph, given an interval representation of digraph as an input. This improves and generalizes Prisner's similar results about interval nest digraphs mentioned above. Moreover, we give an $O((n+m)n)$ time algorithm for the {\sc Min-Kernel} and {\sc Max-Kernel} problems for a superclass of reflexive interval digraphs (here $m$ denotes the number of edges in the digraph other than the self-loops at each vertex). As a consequence, we obtain an improvement over the $O(n^3)$ time algorithm for finding a minimum independent dominating set in cocomparability graphs that was given by Kratsch and Stewart~\cite{kratsch1993domination}. Our algorithm for {\sc Min-Kernel} and {\sc Max-Kernel} problems has a better running time of $O(n^2)$ for adjusted interval digraphs. On the other hand, we show that the problem {\sc Kernel} is NP-complete for point-point digraphs and {\sc Min-Kernel} and {\sc Max-Kernel} problems are APX-hard for point-point digraphs. 
	
\medskip
	
\noindent\textbf{Outline of the paper:} In the remaining part of this section, we give a literature survey on the previous works related to the problems and graph classes of our interest, and also define some of the notation that we use in this paper. In Section~\ref{sec:ordering}, we give our ordering characterization for reflexive interval digraphs. Section~\ref{sec:algorithms} presents the polynomial-time algorithms for the problems that we consider in the class of reflexive interval digraphs. In Section~\ref{sec:npcomplete}, we give a characterization for point-point digraphs followed by the NP-completeness and/or APX-hardness results for point-point digraphs. Section~\ref{sec:conclusion} contains some concluding remarks and proposes some possible directions for further research.
		
%

\subsection{Literature survey}
The problems of computing a maximum independent set and minimum dominating set in undirected graphs are two classic optimization problems in graph theory. As we have noted before, the {\sc Independent-Set} problem in a directed graph coincides with the problem of finding a maximum cardinality independent set of its underlying undirected graph. Also, in order to find a maximum independent set in an undirected graph, one could just orient the edges of the graph in an arbitrary fashion and solve the {\sc Independent-Set} problem on the resulting digraph. Therefore, there is an easy reduction from the problem of computing a maximum independent set in undirected graphs to the {\sc Independent-Set} problem on digraphs and vice versa, implying that these two problems have the same algorithmic complexity. On the other hand, it seems that the the directed analogue of the domination problem is harder than the undirected version, since even though one can find a minimum dominating set in an undirected graph by replacing every edge with symmetric arcs and then using an algorithm for {\sc Dominating-Set} on digraphs, a reduction in the other direction is not known. In particular, a minimum dominating set in the underlying undirected graph of a digraph need not even be a dominating set of the digraph. For example, any vertex of a complete graph is a dominating set of size 1, implying that the problem of finding a minimum cardinality dominating set in a complete graph is trivial, while no polynomial-time algorithm is known to solve the {\sc Dominating-Set} problem for the class of tournaments, which are precisely orientations of complete graphs. Even though domination in tournaments is well studied in the literature~\cite{megiddo1988finding,alon2006dominating,chudnovsky2018domination}, very little is known about the algorithmic complexity of the {\sc Dominating-Set} problem in digraphs. Nevertheless, {\sc Kernel} is a variant of {\sc Dominating-Set} that has gained the attention of researchers over the years. Apart from game theory, the notion of kernel historically played an important role as an approach towards the proof of the celebrated `Strong perfect graph conjecture' (now Strong Perfect Graph Theorem). A digraph $G$ is called \emph{normal} if every clique in $G$ has a kernel (that is, every clique contains a vertex that is an out-neighbor of every other vertex of the clique). Berge and Duchet (see~\cite{boros1996perfect}) 
introduced a notion called \emph{kernel-solvable} graphs, which are undirected graphs for which every normal orientation (symmetric arcs are allowed) of it has a kernel. They conjectured that kernel solvable graphs are exactly the perfect graphs. This conjecture was shown to be true for various special graph classes~\cite{blidia1993kernels,maffray1986kernels,maffray1992kernels}. In general graphs, it was proved by Boros and Gurvich~\cite{boros1996perfect} that perfect graphs are kernel solvable and the converse direction follows from the Strong Perfect Graph Theorem. Kernels are also closely related to \emph{Grundy functions} in digraphs (for a digraph $G=(V,E)$, a non-negative function $f:V\rightarrow \mathbb{N}_{>0}$ is called a Grundy function, if for each vertex $v\in V$, $f(v)$ is the smallest non-negative integer that does not belong to the set $\{ f(u): u\in N^+(v)\}$). Berge~\cite{berge1973graphs} showed that if a digraph has a Grundy function then it has a kernel. Even though the converse is not necessarily true for general digraphs, Berge~\cite{berge1973graphs} proved that every kernel-perfect graph has a Grundy function. It is known that almost every random digraph has a kernel~\cite{de1990kernels}. Kernels, its variants and kernel-perfect graphs are topics that have been extensively studied in the literature, including in the works by Richardson~\cite{richardson1953extension}, Galeana-Sánchez and Neumann-Lara~\cite{galeana1984kernels}, Berge and Duchet~\cite{berge1990recent}, and many more. See~\cite{boros2006perfect} for a detailed survey of results related to kernels. 

Though every normal orientation of a perfect graph has a kernel, the question of finding a kernel has been noted as a challenging problem even in such digraphs. Polynomial-time algorithms for the {\sc Kernel} problem, that also compute a kernel in case one exists, have been obtained for some special graph classes.
K\"onig (see~\cite{haynes1998domination}), 
who was one of the earliest to study domination in digraphs (he called an independent dominating set a `basis of second kind'), proves that every minimal absorbing set of a transitive digraph is a kernel and every kernel in a transitive digraph has the same cardinality. Thus the {\sc Kernel} problem is trivial for transitive digraphs and there is a simple linear time algorithm for the {\sc Min-Kernel} problem in such digraphs. 
The problem of computing a kernel, if one exists, in polynomial time can be solved in digraphs that do not contain odd directed cycles using Richardson's Theorem~\cite{richardson1946weakly}. This implies that this problem is also polynomial-time solvable in directed acyclic graphs. Polynomial-time algorithms for finding a kernel, if one exists, is also known for digraphs that are normal orientations of permutation graphs~\cite{abbas2005polynomial}, Meyniel orientations (an orientation $D$ of $G$ for which every triangle in $D$ has at least two symmetric arcs) of comparability graphs~\cite{abbas2005polynomial}, normal orientations (without symmetric arcs) of claw-free graphs~\cite{pass2020perfect}, normal orientations of chordal graphs~\cite{pass2020perfect} and normal orientations of directed edge graphs (intersection graphs of directed paths in a directed tree)~\cite{pass2020perfect,de2011solving}. For the class of normal orientations of line graphs of bipartite graphs, Maffray~\cite{maffray1992kernels} observed that kernels in such graphs coincide with the stable matchings in the corresponding bipartite graphs. Thus in this graph class, a kernel can be computed in polynomial time using the celebrated algorithm of Gale and Shapely ~\cite{gale1962college} for stable matchings in bipartite graphs. It is shown in~\cite{pass2020perfect} that for any orientation (without symmetric arcs) of circular arc graphs, {\sc kernel} can be solved in polynomial time and a kernel, if one exists, can also be computed in polynomial time. The problem was also solved for the class of interval nest digraphs by Prisner~\cite{prisner1994algorithms}.

In this paper, we study the {\sc Kernel}, {\sc Min-Kernel}, {\sc Max-Kernel}, {\sc Absorbing-Set}, {\sc Dominating-Set}, and {\sc Independent-Set} problems in the class of interval digraphs and its subclasses. Interval digraphs were introduced by Das, Roy, Sen and West~\cite{das1989interval} in 1989. They provided a characterization for the adjacency matrices of interval digraphs and also showed that they are exactly the digraphs formed by the intersection of two Ferrers digraphs whose union is a complete digraph (see~\cite{das1989interval}).
Many subclasses of interval digraphs have attracted the interest of researchers over the years since then. The authors of~\cite{das1989interval} studied the special subclass of interval digraphs called interval point digraphs. If a digraph $G$ has an interval representation in which $T_u$ is a point that lies inside the interval $S_u$ for each vertex $u\in V(G)$, the graph $G$ is said to be an \emph{interval catch digraph}. Even more restrictively, if the point $T_u$ is the left end-point of the interval $S_u$ for each vertex $u$, then the digraph is said to be a \emph{chronological interval digraph}; such digraphs were introduced and characterized in~\cite{das2013recognition}. We would like to note here that interval catch digraphs were defined and studied in the work of Maehara~\cite{maehara1984digraph} that predates the introduction of interval digraphs (the term ``interval digraph'' was used with a different meaning in this work). A forbidden structure characterization and a polynomial time recognition algorithm for interval catch digraphs was presented in~\cite{prisner1989characterization}. Prisner~\cite{prisner1994algorithms} generalized interval catch digraphs to interval nest digraphs and provided a polynomial-time recognition algorithm for interval point digraphs. The class of adjusted interval digraphs were introduced by Feder, Hell, Huang, and Rafiey~\cite{feder2012adjusted}. They showed that the list homomorphism problem for a target digraph $H$ is polynomial-time solvable if $H$ is an adjusted interval digraph and conjecture that if $H$ is not an adjusted interval digraph, then the problem is NP-complete (see~\cite{feder2012adjusted}).

\subsection{Notation}\label{sec:notation}
For a closed interval $I=[x,y]$ of the real line (here $x,y\in\mathbb{R}$ and $x\leq y$), we denote by $l(I)$ the left end-point $x$ of $I$ and by $r(I)$ the right end-point $y$ of $I$. We use the following observation throughout the paper: if $I$ and $J$ are two intervals, then $I\cap J=\emptyset\Leftrightarrow (r(I)<l(J))\vee (r(J)<l(I))$. Given an interval representation of a graph, we can always perturb the endpoints of the intervals slightly to obtain an interval representation of the same graph which has the property that no endpoint of an interval coincides with any other endpoint of an interval. We assume that every interval representation considered in this paper has this property.

Let $G=(V,E)$ be a directed graph. For $u,v\in V(G)$, we say that $u$ is an \emph{in-neighbour} (resp. \emph{out-neighbour}) of $v$ if $(u,v)\in E(G)$ (resp. $(v,u)\in E(G)$). For a vertex $v$ in $G$, we denote by $N^+_G(v)$ and $N^-_G(v)$ the set of out-neighbours and the set of in-neighbours of the vertex $v$ in $G$ respectively. When the graph $G$ under consideration is clear from the context, we abbreviate $N^+_G(v)$ and $N^-_G(v)$ to just $N^+(v)$ and $N^-(v)$ respectively. We denote by $n$ the number of vertices in the digraph under consideration, and by $m$ the number of edges in it not including any self-loops.

For $i,j\in\mathbb{N}$ such that $i\leq j$, let $[i,j]$ denote the set $\{i,i+1,\ldots,j\}$. Let $G$ be a digraph with vertex set $[1,n]$. Then for $i,j\in [1,n]$, we define $N^+_{>j}(i)= N^+(i)\cap [j+1,n]$, $N^-_{>j}(i)= N^-(i)\cap [j+1,n]$, $N^+_{<j}(i)= N^+(i)\cap [1,j-1]$, and $N^-_{<j}(i)= N^-(i)\cap [1,j-1]$.
We denote by $\overline{N^+_{>j}(i)}$ and $\overline{N^-_{>j}(i)}$ the sets $[j+1,n]\setminus N^+_{>j}(i)$ and $[j+1,n]\setminus N^-_{>j}(i)$ respectively.
\section{Ordering characterization} \label{sec:ordering}
We first show that a digraph is a reflexive interval digraph if and only if there is a linear ordering of its vertex set such that none of the structures shown in Figure~\ref{fig:forbidden} are present.
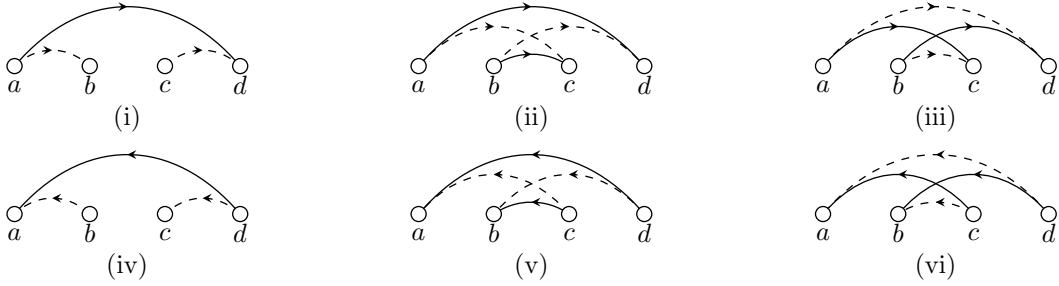
\begin{figure}
\begin{tabular}{p{.3\textwidth}p{.3\textwidth}p{.3\textwidth}}
	\parbox{.3\textwidth}{\centering
	\newcommand{\myptr}{{\arrow{stealth}}}
	\renewcommand{\vertexset}{(a,0,0),(b,1,0),(c,2,0),(d,3,0)}
	\renewcommand{\edgeset}{(a,b,,,8,dashed),(a,d,,,30),(c,d,,,8,dashed)}
	\renewcommand{\defradius}{.1}
	\renewcommand{\isdirected}{myptr}
	\begin{tikzpicture}
	\drawgraph
	\node [below=2] at (\xy{a}) {$a$};
	\node [below] at (\xy{b}) {$b$};
	\node [below=2] at (\xy{c}) {$c$};
	\node [below] at (\xy{d}) {$d$};
	\end{tikzpicture}}
	& 
	\parbox{.3\textwidth}{\centering
	\newcommand{\myptr}{{\arrow{stealth}}}
	\renewcommand{\vertexset}{(a,0,0),(b,1,0),(c,2,0),(d,3,0)}
	\renewcommand{\edgeset}{(a,c,,,20,dashed),(a,d,,,30),(b,c,,,6),(b,d,,,20,dashed)}
	\renewcommand{\defradius}{.1}
	\renewcommand{\isdirected}{myptr}
	\begin{tikzpicture}
	\drawgraph
	\node [below=2] at (\xy{a}) {$a$};
	\node [below] at (\xy{b}) {$b$};
	\node [below=2] at (\xy{c}) {$c$};
	\node [below] at (\xy{d}) {$d$};
	\end{tikzpicture}}
	&
	\parbox{.3\textwidth}{\centering
	\newcommand{\myptr}{{\arrow{stealth}}}
	\renewcommand{\vertexset}{(a,0,0),(b,1,0),(c,2,0),(d,3,0)}
	\renewcommand{\edgeset}{(a,d,,,30,dashed),(a,c,,,20),(b,d,,,20),(b,c,,,6,dashed)}
	\renewcommand{\defradius}{.1}
	\renewcommand{\isdirected}{myptr}
	\begin{tikzpicture}
	\drawgraph
	\node [below=2] at (\xy{a}) {$a$};
	\node [below] at (\xy{b}) {$b$};
	\node [below=2] at (\xy{c}) {$c$};
	\node [below] at (\xy{d}) {$d$};
	\end{tikzpicture}}\\
	\parbox{.3\textwidth}{\centering (i)}&\parbox{.3\textwidth}{\centering (ii)}&\parbox{.3\textwidth}{\centering (iii)}\\
	\parbox{.3\textwidth}{\centering
	\newcommand{\myptr}{{\arrow{stealth}}}
	\renewcommand{\vertexset}{(a,0,0),(b,1,0),(c,2,0),(d,3,0)}
	\renewcommand{\edgeset}{(b,a,,,-8,dashed),(d,a,,,-30),(d,c,,,-8,dashed)}
	\renewcommand{\defradius}{.1}
	\renewcommand{\isdirected}{myptr}
	\begin{tikzpicture}
	\drawgraph
	\node [below=2] at (\xy{a}) {$a$};
	\node [below] at (\xy{b}) {$b$};
	\node [below=2] at (\xy{c}) {$c$};
	\node [below] at (\xy{d}) {$d$};
	\end{tikzpicture}}
	& 
	\parbox{.3\textwidth}{\centering 
	\newcommand{\myptr}{{\arrow{stealth}}}
	\renewcommand{\vertexset}{(a,0,0),(b,1,0),(c,2,0),(d,3,0)}
	\renewcommand{\edgeset}{(c,a,,,-20,dashed),(d,a,,,-30),(c,b,,,-6),(d,b,,,-20,dashed)}
	\renewcommand{\defradius}{.1}
	\renewcommand{\isdirected}{myptr}
	\begin{tikzpicture}
	\drawgraph
	\node [below=2] at (\xy{a}) {$a$};
	\node [below] at (\xy{b}) {$b$};
	\node [below=2] at (\xy{c}) {$c$};
	\node [below] at (\xy{d}) {$d$};
	\end{tikzpicture}}
	&
	\parbox{.3\textwidth}{\centering
	\newcommand{\myptr}{{\arrow{stealth}}}
	\renewcommand{\vertexset}{(a,0,0),(b,1,0),(c,2,0),(d,3,0)}
	\renewcommand{\edgeset}{(d,a,,,-30,dashed),(c,a,,,-20),(d,b,,,-20),(c,b,,,-6,dashed)}
	\renewcommand{\defradius}{.1}
	\renewcommand{\isdirected}{myptr}
	\begin{tikzpicture}
	\drawgraph
	\node [below=2] at (\xy{a}) {$a$};
	\node [below] at (\xy{b}) {$b$};
	\node [below=2] at (\xy{c}) {$c$};
	\node [below] at (\xy{d}) {$d$};
	\end{tikzpicture}}\\
	\parbox{.3\textwidth}{\centering (iv)}&\parbox{.3\textwidth}{\centering (v)}&\parbox{.3\textwidth}{\centering (vi)}
\end{tabular}
\caption{Forbidden structures for reflexive interval digraphs (possibly $b=c$ in (i), (ii), (iv) and (v)). A dashed arc from $u$ to $v$ indicates the absence of the edge $(u,v)$ in the graph. Note that the vertices are assumed to have self-loops since a vertex without a self-loop is itself forbidden in a reflexive interval digraph.}\label{fig:forbidden}
\end{figure}

\begin{theorem}\label{thm:ordering}
A digraph $G$ is a reflexive interval digraph if and only if $V(G)$ has an ordering $<$ in which for any $a,b,c,d\in V(G)$ such that $a<b<c<d$, none of the structures in Figure~\ref{fig:forbidden} occur ($b$ and $c$ can be the same vertex in (i), (ii), (iv), (v) of Figure~\ref{fig:forbidden}). 
\end{theorem}
\begin{proof}
Let $G$ be a reflexive interval digraph with an interval representation $\{(S_v,T_v):v\in V(G)\}$. For any vertex $v\in V(G)$, let $x_v$ be the left most end point of the interval $S_v\cap T_v$(which is well defined as $G$ is a reflexive interval digraph). Let $<$ be an ordering of $V(G)$ with respect to the increasing order of the points $x_v$. Now we can verify that structures in Figure~\ref{fig:forbidden} are forbidden with respect to the order $<$.

Suppose not. Let $a<b<c<d$ be such that of Figure~\ref{fig:forbidden}(i). Then $a<b$, $c<d$ and $(a,b), (c,d)\notin E(G)$ implies that $r(S_a)<l(T_b)$ and $r(S_c)< l(T_d)$. Since $b\leq c$, we also have that $l(T_b)\leq r(S_c)$. Combining these observations we then have that $r(S_a)<l(T_d)$, which further implies that $(a,d)\notin E(G)$, which is a contradiction to Figure~\ref{fig:forbidden}(i). Let $a<b<c<d$ be such that of Figure~\ref{fig:forbidden}(ii). Then $a<c$, $b<d$ and $(a,c), (b,d)\notin E(G)$ implies that $r(S_a)<l(T_c)$ and $r(S_b)< l(T_d)$. Since $(a,d)\in E(G)$, we also have that $l(T_d)<r(S_a)$. Combining these observations we then have, $r(S_b)<l(T_c)$ implying that $(b,c)\notin E(G)$, which is a contradiction to Figure~\ref{fig:forbidden}(ii). Suppose that $a<b<c<d$ be such that of Figure~\ref{fig:forbidden}(iii). Then $(a,c), (b,d)\in E(G)$ implies that $l(T_c)<r(S_a)$ and $l(T_d)<r(S_b)$. Since $a<d$, $(a,d)\notin E(G)$, we also have that $r(S_a)<l(T_d)$. Combining these observations, we then have $l(T_c)<r(S_b)$. Since $b<c$, this implies that $(b,c)\in E(G)$, which is a contradiction to Figure~\ref{fig:forbidden}(iii). Since we arrive at a contradiction in every case, we can conclude that none of the structures in Figures~\ref{fig:forbidden}(i), (ii) or (iii) can be present. Similarly, by interchanging the roles of source and destination intervals in the above proof, we can also prove that none of the structures in Figures~\ref{fig:forbidden}(iv), (v) or (vi) can be present with respect to the ordering $<$.

Conversely, assume that $<$ is an ordering of $V(G)$ for which the structures in Figure~\ref{fig:forbidden} are absent. Let $n=|V(G)|$. We can assume that $V(G)=[1,n]$ and that $<$ is the ordering $(1,2,\ldots,n)$. First, we note the following observation.
\begin{observation}\label{subset}
For any two vertices $i,j$ such that $i<j$, we have the following:
\begin{myenumerate}
	\item~\label{it1} either $N^+_{>j}(i)\subseteq N^+_{>j}(j)$ or $N^+_{>j}(j)\subseteq N^+_{>j}(i)$ and
	\item~\label{it2} either $N^-_{>j}(i)\subseteq N^-_{>j}(j)$ or $N^-_{>j}(j)\subseteq N^-_{>j}(i)$.
\end{myenumerate}
\end{observation}
\begin{proof}
Suppose not. Due to the symmetry between~\ref{it1} and~\ref{it2}, we prove only the case where~\ref{it1} is not true. Then there exists two distinct vertices $x_i,x_j \in \{j+1,\ldots,n\}$ such that $x_i\in N^+_{>j}(i)\setminus N^+_{>j}(j)$ and $x_j\in N^+_{>j}(j)\setminus N^+_{>j}(i)$. Now if $x_i<x_j$, then the vertices $i<j<x_i<x_j$ form Figure~\ref{fig:forbidden}(iii) which is forbidden and if $x_j<x_i$, then the vertices $i<j<x_j<x_i$ form Figure~\ref{fig:forbidden}(ii) which is also forbidden. As we have a contradiction in both the cases, we are done. 
\end{proof}
\medskip

We now define for each $i\in\{1,2,\ldots,n\}$, a pair of intervals $(S_i,T_i)$ as follows.
For each $i\in \{1,2,\ldots, n\}$, let $$y_i=\begin{array}{lll}\begin{cases}
\min\overline{N^+_{>i}(i)}, & \text{ if } \overline{N^+_{>i}(i)}\neq \emptyset \\
n+1, & \text{ otherwise} 
\end{cases}&\hspace{.5in}\text{and}\hspace{.5in}&z_i= |N^+_{>y_i}(i)|.\end{array}$$

Define, $r(S_i)=y_i-1+\frac{z_i}{n+1}$ and $l(T_i)= \min\big(\{i\}\cup\{r(S_j):j\in N^-_{<i}(i)\}\big)$.

Similarly let, $$y'_i= \begin{array}{lll}\begin{cases}
\min\overline{N^-_{>i}(i)}, & \text{ if } \overline{N^-_{>i}(i)}\neq \emptyset \\
n+1, & \text{ otherwise} 
\end{cases}&\hspace{.5in}\text{and}\hspace{.5in}&z'_i= |N^-_{>y'_i}(i)|.\end{array}$$

Define, $r(T_i)=y'_i-1+\frac{z'_i}{n+1}$ and $l(S_i)= \min\big(\{i\}\cup\{r(T_j):j\in N^+_{<i}(i)\}\big)$.\medskip

Note that for each vertex $i\in V(G)$, by the above definition of intervals corresponding to $i$, we have that the point $i\in S_i\cap T_i$, $y_i-1\leq r(S_i)<y_i$ and $y'_i-1 \leq r(T_i)<y'_i$. Also for any two vertices $i,j$ such that $y_i=y_j=p$, we have by Observation~\ref{subset} that $r(S_i)\leq r(S_j)$ if and only if $z_i\leq z_j$ if and only if $N^+_{>p}(i)\subseteq N^+_{>p}(j)$. Similarly, for any two vertices $i,j$ such that $y'_i=y'_j=q$, we have that $r(T_i)\leq r(T_j)$ if and only if $z'_i\leq z'_j$ if and only if $N^-_{>q}(i)\subseteq N^-_{>q}(j)$.\medskip 
 
Now we have to prove that $E(G)= \{(i,j):S_i\cap T_j \neq \emptyset\}$. Let $(i,j)\in E(G)$ be such that $i<j$. If $j<y_i$, then we have $l(S_i)\leq i<j\leq r(S_i)$, implying that $S_i\cap T_j \neq \emptyset$ (recall that $j\in S_j\cap T_j$). Suppose that $y_i<j$. Then we have $l(T_j)\leq r(S_i)<y_i <j<r(T_j)$ implying that $S_i\cap T_j \neq \emptyset$. In a similar way, by interchanging the roles of source and destination intervals and that of $i$ and $j$, we can also prove that: if $(i,j)\in E(G)$ be such that $j<i$, then $S_i \cap T_j\neq \emptyset$. On the other hand, suppose that $(i,j)\notin E(G)$, where $i<j$. Clearly, then $y_i\leq j$. For the sake of contradiction assume that $S_i\cap T_j \neq \emptyset$. Since $r(S_i)< y_i$, this is possible only if $l(T_j)\leq r(S_i)<y_i\leq j$. Thus, $l(T_j)<j$, which implies by the definition of intervals that $N^-_{<j}(j)\neq\emptyset$. Let $k\in N^-_{<j}(j)$ such that $r(S_k)=\min\{r(S_l):l\in N^-_{<j}(j)\}$. Since $(k,j)\in E(G)$ and $r(S_k)=l(T_j)<j$, we can conclude by the definition of $r(S_k)$ that $y_k<j$. 
 Suppose that $y_i=y_k=p$. Then, since $r(S_k)=l(T_j)\leq r(S_i)$, we can conclude by our earlier observation that $N^+_{>p}(k)\subseteq N^+_{>p}(i)$, which contradicts the fact that $j\in N^+_{>p}(k)\setminus N^+_{>p}(i)$. We can thus infer that $y_i\neq y_k$. This together with the fact that $y_k-1\leq r(S_k)= l(T_j)<y_i$, implies that $y_k<y_i$. Suppose that $y_k\leq i$, then $k<y_k\leq i<j$, $(k,j)\in E(G)$, and $(k,y_k),(i,j)\notin E(G)$, which gives us Figure~\ref{fig:forbidden}(i), which is a contradiction. Therefore we can assume that $i<y_k$, which further implies that $(i,y_k)\in E(G)$ (recall that $y_k<y_i$). Now we have $y_k\in N^+_{>\max\{i,k\}}(i)\setminus N^+_{>\max\{i,k\}}(k)$ and $j\in N^+_{>\max\{i,k\}}(k)\setminus N^+_{>\max\{i,k\}}(i)$, which contradicts Observation~\ref{subset}. As we arrive at a contradiction in every cases, we can conclude that $S_i\cap T_j = \emptyset$. The case where $(i,j)\notin E(G)$ such that $j<i$ is symmetric.
\end{proof}
Now we define the following.
\begin{definition}[DUF-ordering]
A \emph{directed umbrella-free ordering} (or in short a \emph{DUF-ordering}) of a digraph $G$ is an ordering of $V(G)$ satisfying the following properties for any three distinct vertices $i<j<k$:
\begin{myenumerate}
	\item \label{out} if $(i,k)\in E(G)$, then either $(i,j)\in E(G)$ or $(j,k)\in E(G)$, and
	\item \label{in} if $(k,i)\in E(G)$, then either $(k,j)\in E(G)$ or $(j,i)\in E(G)$. 
\end{myenumerate} 
\end{definition}

\begin{definition}[DUF-digraph]
A digraph $G$ is a \emph{directed umbrella-free digraph} (or in short a \emph{DUF-digraph}) if it has a DUF-ordering.
\end{definition}

Then the following corollary is an immediate consequence of Theorem~\ref{thm:ordering}.
\begin{corollary}\label{cor:duf}
Every reflexive interval digraph is a DUF-digraph.
\end{corollary}
\medskip
Let $G$ be an undirected graph. We define the \emph{symmetric digraph of $G$} to be the digraph obtained by replacing each edge of $G$ by symmetric arcs.
	
\medskip

The following characterization of cocomparability graphs was first given by Damaschke~\cite{damaschke1990}.
\begin{theorem}[\cite{damaschke1990}]\label{thm:cocomp}
An undirected graph $G$ is a cocomparability graph if and only if there is an ordering $<$ of $V(G)$ such that for any three vertices $i<j<k$, if $ik\in E(G)$, then either $ij\in E(G)$ or $jk\in E(G)$.
\end{theorem}
Then we have the following corollary.

\begin{corollary} \label{cor:DUFcc}
The underlying undirected graph of every DUF-digraph is a cocomparablity graph. 
\end{corollary}

Note that there exist digraphs which are not DUF-digraphs but their underlying undirected graphs are cocomparability (for example, a directed triangle with edges $(a,b),(b,c)$ and $(c,a)$). But we can observe that the class of underlying undirected graphs of DUF-digraphs is precisely the class of cocomparability graphs, since it follows from Theorem~\ref{thm:cocomp} that symmetric digraph of any cocomparability graph is a DUF-digraphs. In contrast, the class of underlying undirected graphs of reflexive interval digraphs forms a strict subclass of cocomparability graphs. We prove this by showing that no directed graph that has $K_{3,3}$ as its underlying undirected graph can be a reflexive interval digraph ($K_{3,3}$ can easily be seen to be a cocomparability graph). This would also imply by Corollary~\ref{cor:duf} that the class of reflexive interval digraphs forms a strict subclass of DUF-digraphs.

\begin{theorem}\label{thm:K_33free}
The underlying undirected graph of a reflexive interval digraph cannot contain $K_{3,3}$ as an induced subgraph.
\end{theorem}
\begin{proof}
Since the class of reflexive interval digraphs is closed under taking induced subgraphs, it is enough to prove that the underlying undirected graph of a reflexive interval digraph cannot be $K_{3,3}$. Let $H$ be an undirected graph. An ordering $<$ of $V(H)$ is said to be a \emph{special umbrella-free} ordering of $H$, if for any four distinct vertices $a,b,c,d\in V(G)$ such that $a< b< c< d$, $ad\in E(H)$ implies that either $ab\in E(H)$ or $cd\in E(H)$. Let $G$ be any reflexive interval digraph. By Theorem~\ref{thm:ordering}, we have that $V(G)$ has an ordering such that none of the structures in Figure~\ref{fig:forbidden} are present. It follows that this ordering is also a special umbrella-free ordering of the underlying undirected graph of $G$. Therefore we can conclude that the underlying undirected graph of any reflexive interval digraph has a special umbrella-free ordering. We claim that $K_{3,3}$ does not have a special umbrella-free ordering, which then implies the theorem.

Let $A$ and $B$ denote the two partite sets of the bipartite graph $K_{3,3}$. Suppose for the sake of contradiction that $K_{3,3}$ has a special umbrella-free ordering $<\colon (v_1,v_2,\ldots,v_6)$. Suppose that $v_1$ and $v_6$ belong to different partite sets of $K_{3,3}$. Without loss of generality, we can assume that $v_1\in A$ and $v_6\in B$. This implies that there cannot exist vertices $v_i,v_j\in \{v_2,v_3,v_4,v_5\}$ such that $v_i< v_j$, $v_i\in A$ and $v_j\in B$, as otherwise we have $v_1< v_i< v_j< v_6$, $v_1v_6\in E(K_{3,3})$, and $v_1v_i,v_jv_6 \notin E(K_{3,3})$, which contradicts the fact that $< $ is a special umbrella-free ordering. This further implies that $v_2,v_3\in B$ and $v_4,v_5\in A$. Then we have $v_2< v_3< v_4< v_5$, $v_2v_5\in E(K_{3,3})$, and $v_2v_3,v_4v_5\notin E(K_{3,3})$, which is again a contradiction. Therefore we can assume that $v_1$ and $v_6$ belong to the same partite set of $K_{3,3}$. Without loss of generality, we can assume that $v_1,v_6\in A$. Now if $v_2\in A$, then we have $v_3,v_4,v_5\in B$. Then we have $v_1< v_2< v_3< v_4$, $v_1v_4\in E(K_{3,3})$, and $v_1v_2,v_3v_4\notin E(K_{3,3})$, which is again a contradiction. This implies that $v_2\in B$. Now if there exists a vertex $x\in\{v_4,v_5\}\cap A$, then we have $v_3\in B$, in which case we have $v_2< v_3< x< v_6$, $v_2v_6\in E(K_{3,3})$, and $v_2v_3,xv_6\notin E(K_{3,3})$, which is again a contradiction. Therefore we can assume that $v_4,v_5\in B$, implying that $v_3\in A$. Then we have $v_1< v_3< v_4< v_5$, $v_1v_5\in E(K_{3,3})$, and $v_1v_3,v_4v_5 \notin E(K_{3,3})$, which is again a contradiction. This shows that $K_{3,3}$ has no special umbrella-free ordering, thereby proving the theorem.
\end{proof}
 Prisner~\cite{prisner1994algorithms} proved the following.
\begin{theorem}[\cite{prisner1994algorithms}]\label{thm:prisnerweaklytriang}
The underlying undirected graphs of interval nest digraphs are weakly chordal graphs.
\end{theorem}
By Corollaries~\ref{cor:duf}, \ref{cor:DUFcc} and Theorem~\ref{thm:K_33free}, we can conclude that the underlying undirected graphs of reflexive interval digraphs are $K_{3,3}$-free cocomparability graphs. This strengthens Theorem~\ref{thm:prisnerweaklytriang}, since now we have that the underlying undirected graphs of interval nest digraphs are $K_{3,3}$-free weakly chordal cocomparability graphs. 
\section{Algorithms for reflexive interval digraphs} \label{sec:algorithms}
Here we explore the three different problems defined in Section~\ref{sec:intro} in the class of reflexive interval digraphs.

Let $G$ be a reflexive interval digraph. Note that any induced subdigraph of $G$ is also a reflexive interval digraph and that the ``reversal'' of $G$ --- the digraph obtained by replacing each edge $(u,v)$ of $G$ by $(v,u)$) --- is also a reflexive interval digraph. Since in any digraph, a set $S$ is an absorbing set (resp. kernel) if and only if it is a dominating set (resp. solution) in its reversal, this means that any algorithm that solves {\sc Absorbing-Set} (resp. {\sc Kernel}) problem for the class of reflexive interval digraphs can also be used to solve the {\sc Dominating-Set} (resp. {\sc Solution}) problem on an input reflexive interval digraph. Therefore, in the sequel, we only study the {\sc Absorbing-Set} and {\sc Kernel} problems on reflexive interval digraphs.
\subsection{Kernel}
We use the following result of Prisner that is implied by Theorem~4.2 of~\cite{prisner1994algorithms}.

\begin{theorem}[\cite{prisner1994algorithms}]\label{thm:prisner}
Let $\mathcal{C}$ be a class of digraphs that is closed under taking induced subgraphs. If in every graph $G\in\mathcal{C}$, there exists a vertex $z$ such that for every $y\in N^-(z)$, $N^+(z)\setminus N^-(z)\subseteq N^+(y)$, then the class $\mathcal{C}$ is kernel-perfect.
\end{theorem}

\begin{lemma}\label{lem:z}
Let $G$ be a reflexive interval digraph $G$ with interval representation $\{(S_u,T_u)\}_{u\in V(G)}$. Let $z$ be the vertex such that $r(S_z)=\min\{r(S_v):v\in V(G)\}$. Then for every $y\in N^-(z)$, $N^+(z)\setminus N^-(z)\subseteq N^+(y)$.
\end{lemma}
\begin{proof}
Let $x\in N^+(z)\setminus N^-(z)$ and $y\in N^-(z)$. We have to prove that $x\in N^+(y)$.
By the choice of $z$, we have that $r(S_x), r(S_y)> r(S_z)$. As $S_z\cap T_z\neq\emptyset$ (since $G$ is reflexive interval digraph), we have $l(T_z)<r(S_z)$. Combining with the previous inequality, we have $l(T_z)<r(S_x)$. As $x\notin N^-(z)$, it then follows that $l(S_x)>r(T_z)$. Since $y\in N^-(z)$, we have that $l(S_y)<r(T_z)$. We now have that $l(S_y)<l(S_x)$. As $l(S_x)<r(T_x)$ this further implies that $l(S_y)<r(T_x)$. Now if $x\notin N^+(y)$, it should be the case that $l(T_x)>r(S_y)>r(S_z)$ which is a contradiction to the fact that $x\in N^+(z)$.
\end{proof}

Since reflexive interval digraphs are closed under taking induced subgraphs, by Theorem~\ref{thm:prisner} and Lemma~\ref{lem:z}, we have the following.

\begin{theorem}\label{thm:kernelperfect}
Reflexive interval digraphs are kernel-perfect.
\end{theorem}  

It follows from the above theorem that the decision problem {\sc Kernel} is trivial on reflexive interval digraphs. As explained below, we can also compute a kernel in a reflexive interval digraph efficiently, if an interval representation of the digraph is known.

Let $G$ be a reflexive interval digraph with an interval representation $\{(S_u,T_u)\}_{u\in V(G)}$. Let $G_0=G$ and $z_0$ be the vertex in $G$ such that $r(S_{z_0})=\min\{r(S_v):v\in V(G)\}$. For $i\geq 1$, recursively define $G_i$ to be the induced subdigraph of $G$ with $V(G_i)=V(G_{i-1})\setminus (\{z_{i-1}\}\cup N^-(z_{i-1}))$ and if $V(G_i)\neq \emptyset$, define $z_i$ to be the vertex such that $r(S_{z_i})=\min\{r(S_v):v\in V(G_i)\}$. Let $t$ be smallest integer such that $V(G_{t+1})=\emptyset$. Note that this implies that $V(G_t)=\{z_t\}\cup N^-_{G_t}(z_t)$. Clearly $t\leq n$ and $r(S_{z_0})<r(S_{z_1})<\cdots<r(S_{z_t})$. By Lemma~\ref{lem:z}, we have that for each $i\in \{1,2,\ldots,t\}$, $z_i$ has the following property: for any $y\in N_{G_i}^-(z_i)$ we have $N_{G_i}^+(z_i)\setminus N_{G_i}^-(z_i)\subseteq N_{G_i}^+(y)$.
\medskip

We now recursively define a set $K_i\subseteq V(G_i)$ as follows: Define $K_t=\{z_t\}$. For each $i\in \{t-1,t-2,\ldots,0\}$,
$$ K_i = \begin{cases}
\{z_i\}\cup K_{i+1} & \text{ if } (z_i,z_j)\notin E(G), \text{ where } j=\min\{l:z_l\in K_{i+1}\} \\
K_{i+1} & \text{ otherwise.}
\end{cases}$$
\begin{lemma}\label{lem:kerreflexive}
For each $i\in \{1,2,\ldots,t\}$, $K_i$ is a kernel of $G_i$.
\end{lemma}
\begin{proof}
We prove this by reverse induction on $i$. The base case where $K_t=\{z_t\}$ is trivial since $V(G_t)= \{z_t\}\cup N^-_{G_t}(z_t)$. Assume that the hypothesis is true for all $j$ such that $j>i$. If $K_i=K_{i+1}$ then it implies that there exists $z_j\in K_{i+1}$ such that $z_j\in N^+(z_i)$. Further as $z_j\in V(G_{i+1})=V(G_i)\setminus (\{z_i\}\cup N_{G_i}^-(z_i))$, we have that $z_j\in N_{G_i}^+(z_i)\setminus N_{G_i}^-(z_i)$. Let $y\in  N_{G_i}^-(z_i)$.
Since $N_{G_i}^+(z_i)\setminus N_{G_i}^-(z_i)\subseteq N_{G_i}^+(y)$, we then have that $y\in N_{G_i}^-(z_j)$. Thus $N_{G_i}^-(z_i)\subseteq N_{G_i}^-(z_j)$. As $z_i\in N^-(z_j)$, it follows that every vertex in $V(G_i)\setminus V(G_{i+1})=\{z_i\}\cup N^-_{G_i}(z_i)$ is an in-neighbor of $z_j$. We can now use the induction hypothesis to conclude that $K_i=K_{i+1}$ is a kernel of $G_i$. On the other hand, if $K_i=\{z_i\}\cup K_{i+1}$, then it should be the case that $(z_i,z_j)\notin E(G)$ where $j=\min\{l:z_l\in K_{i+1}\}$.
Now consider any $z_l\in K_{i+1}$ where $z_l\neq z_j$. By definition of $j$, we have $l>j$. If $(z_i,z_l)\in E(G)$, then as $r(S_{z_i})<r(S_{z_j})<r(S_{z_l})$, it should be the case that $l(T_{z_l})<r(S_{z_i})<r(S_{z_j})<r(S_{z_l})$.
We also have $l(S_{z_j})<l(S_{z_l})$ as otherwise $S_{z_j}\subseteq S_{z_l}$, implying that $S_{z_l}\cap T_{z_j}\neq\emptyset$, contradicting the fact that $(z_l,z_j)\notin E(G)$ (as $z_l$ and $z_j$ both belong to $K_{i+1}$, which by the induction hypothesis is a kernel of $G_{i+1}$). Since $r(T_{z_l})>l(S_{z_l})>l(S_{z_j})$ and $r(S_{z_j})>l(T_{z_l})$, we now have that $S_{z_j}\cap T_{z_l}\neq\emptyset$, which is a contradiction to the fact that  $(z_j,z_l)\notin E(G)$ (as $z_j,z_l\in K_{i+1}$, which by the induction hypothesis is a kernel of $G_{i+1}$). Thus no vertex in $K_{i+1}$ can be an out-neighbor of $z_i$. By definition of $G_{i+1}$, no vertex in $G_{i+1}$, and hence no vertex in $K_{i+1}$, can be an in-neighbor of $z_i$. Then we have by the induction hypothesis that $K_i=\{z_i\}\cup K_{i+1}$ is an independent set. Since the only vertices in $V(G_i)\setminus V(G_{i+1})$ are $\{z_i\}\cup N^-_{G_i}(z_i)$, and $K_{i+1}$ is an absorbing set of $G_{i+1}$ by the induction hypothesis, we can conclude that $K_i=\{z_i\}\cup K_{i+1}$ is an absorbing set of $G_i$. Therefore $K_i$ is a kernel of $G_i$.
\end{proof}
By the above lemma, we have that $K_0$ is a kernel of $G$. We can now construct an algorithm that computes a kernel in a reflexive interval digraph $G$, given an interval representation of it. We assume that the interval representation of $G$ is given in the form of a list of left and right endpoints of intervals corresponding to the vertices. We can process this list from left to right in a single pass to compute the list of vertices $z_0,z_1,\ldots,z_t$ in $O(n+m)$ time. We then process this new list from right to left in a single pass to generate a set $K$ as follows: initialize $K=\{z_t\}$ and for each $i\in\{t-1,t-2,\ldots,0\}$, add $z_i$ to $K$ if it is not an in-neighbor of the last vertex that was added to $K$. Clearly, the set $K$ can be generated in $O(n+m)$ time. It is easy to see that $K=K_0$ and therefore by Lemma~\ref{lem:kerreflexive}, $K$ is a kernel of $G$. Thus, we have the following theorem.

\begin{theorem}\label{thm:algkerreflexive}
A kernel of a reflexive interval digraph can be computed in linear-time, given an interval representation of the digraph as input.
\end{theorem}
 
The linear-time algorithm described above is an improvement and generalization of the Prisner's result that, interval nest digraphs and their reversals are kernel-perfect, and a kernel can be found in these graphs in time $O(n^2)$ if a representation of the graph is given~\cite{prisner1994algorithms}. 

Now it is interesting to note that even for some kernel perfect digraphs with a polynomial-time computable kernel, the problems {\sc Min-Kernel} and {\sc Max-Kernel} turn out to be NP-complete. The following remark provides an example of such a class of digraphs. 
\begin{remark}\label{rem:minkernel}
Let $\mathcal{C}$ be the class of symmetric digraphs of undirected graphs. Note that the class $\mathcal{C}$ is kernel-perfect, as for any $G\in \mathcal{C}$ the kernels of  the digraph $G$ are exactly the independent dominating sets of its underlying undirected graph. Note that any maximal independent set of an undirected graph is also an independent dominating set of it. Therefore, as a maximal independent set of any undirected graph can be found in linear-time, the problem {\sc Kernel} is linear-time solvable for the class $\mathcal{C}$. On the other hand, note that the problems {\sc Min-Kernel} and {\sc Max-Kernel} for the class $\mathcal{C}$ is equivalent to the problems of finding a minimum cardinality independent dominating set and a maximum cardinality independent set for the class of undirected graphs, respectively. Since the latter problems are NP-complete for the class of undirected graphs, we have that the problems {\sc Min-Kernel} and {\sc Max-Kernel} are NP-complete in $\mathcal{C}$.
\end{remark} 

Note that unlike the class of reflexive interval digraphs, the class of DUF-digraphs are not kernel perfect. Figure~\ref{fig:nokernel} provides an example for a DUF-digraph that has no kernel. Since that graph is a semi-complete digraph (i.e. each pair of vertices is adjacent), and every vertex has an out-neighbor which is not its in-neighbor, it cannot have a kernel. The ordering of the vertices of the graph that is shown in the figure can easily be verified to be a DUF-ordering.

\begin{figure}
\newcommand{\myptr}{{\arrow{stealth}}}
\renewcommand{\vertexset}{(a,0,0),(b,1,0),(c,2,0),(d,3,0)}
\renewcommand{\edgeset}{(a,b,,,10),(b,a,,,10)(c,a,,,22)(a,d,,,30),(b,c,,,10),(d,b,,,22),(c,d,,,10),(d,c,,,10)}
\renewcommand{\defradius}{.1}
\renewcommand{\isdirected}{myptr}
\begin{center}
\begin{tikzpicture}
\drawgraph
\end{tikzpicture}
\end{center}
\caption{Example of a DUF-digraph that has no kernel.}\label{fig:nokernel}
\end{figure}
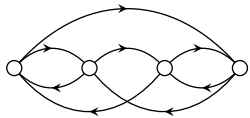

In contrast to Remark~\ref{rem:minkernel}, even though DUF-digraphs may not have kernels, we show in the next section that the problems {\sc kernel} and {\sc Min-Kernel} can be solved in polynomial time in the class of DUF-digraphs. In fact we give a polynomial time algorithm that, given a DUF-digraph $G$ with a DUF-ordering as input, either finds a minimum sized kernel in $G$ or correctly concludes that $G$ does not have a kernel.
\subsection{Minimum sized kernel}
Let $G$ be a DUF-digraph with vertex set $[1,n]$. We assume without loss of generality that $<\colon (1,2,\ldots,n)$ is a DUF-ordering of $G$. Let $i\in\{1,2,\ldots,n\}$. In this section, we shorten $N^+_{>i}(i)$ and $N^-_{>i}(i)$ to $N^+_>(i)$ and $N^-_>(i)$ respectively for ease of notation. We further define $N_>(i)=N^+_>(i)\cup N^-_>(i)$ and define $\overline{N^+_>(i)}$, $\overline{N^-_>(i)}$, $\overline{N_>(i)}$ to be $[i+1,n]\setminus N^+_>(i)$, $[i+1,n]\setminus N^-_>(i)$, $[i+1,n]\setminus N_>(i)$ respectively.

For any vertex $i\in \{1,2,\ldots,n\}$, let $P_i=\{j:j\in \overline{N_>(i)} \text{ such that } [i+1,j-1]\subseteq N^-(i)\cup N^-(j) \}$ and let $G[i,n]$ denote the subgraph induced in $G$ by the set $[i,n]$. Note that we consider $[i+1,j-1]=\emptyset$, if $j=i+1$. For a collection of sets $\mathcal{S}$, we denote by $\mathrm{Min}(\mathcal{S})$ an arbitrarily chosen set in $\mathcal{S}$ of the smallest cardinality. For each $i\in\{1,2,\ldots,n\}$, we define a set $K(i)$ as follows. Here, when we write $K(i)=\infty$, we mean that the set $K(i)$ is undefined.
$$K(i) =  \begin{cases} \{i\}, & \text{if } N^-_>(i) = \{i+1,\ldots,n\} \\
\{i\}\cup \mathrm{Min}\{K(j)\neq\infty:j\in P_i \} , & \text{if } P_i\neq \emptyset \text{ and }\exists j\in P_i \text{ such that } K(j)\neq \infty\\ 
\infty, & \text{otherwise} \end{cases}$$ 

Note that it follows from the above definition that $K(n)=\{n\}$.
For each $i\in \{1,2,\ldots,n\}$, let $OPT(i)$ denote a minimum sized kernel of $G[i,n]$ that also contains $i$. If $G[i,n]$ has no kernel that contains $i$, then we say that $OPT(i)=\infty$. We then have the following lemma.
\begin{lemma}\label{lem:kernel}
The following hold.
\begin{myenumerate}
\item \label{kernel} If $K(i)\neq \infty$, then $K(i)$ is a kernel of $G[i,n]$ that contains $i$, and
\item \label{kiy} if $OPT(i)\neq\infty$, then $K(i)\neq\infty$ and $|K(i)|=|OPT(i)|$.
\end{myenumerate}
\end{lemma}
\begin{proof}
\ref{kernel} We prove this by the reverse induction on $i$. Suppose that $K(i)\neq \infty$. The base case where $i=n$ is trivially true. Assume that the hypothesis is true for every $j>i$. It is clear from the definition of $K(i)$ that $i\in K(i)$. If $K(i)= \{i\}$, then it should be the case that $N^-_>(i) = \{i+1,\ldots,n\}$, implying that the set $K(i)= \{i\}$, is both an independent set and an absorbing set in $G[i,n]$, and we are done. Otherwise, $K(i)=\{i\}\cup K(j)$ for some $j\in P_i$ such that $K(j)\neq \infty$. By the definition of $P_i$, we have that $j\in \overline{N_>(i)}$ and $[i+1,j-1]\subseteq N^-(i)\cup N^-(j)$. Since $j>i$, we have by the induction hypothesis that $K(j)$ is an independent and absorbing set in $G[j,n]$. Suppose that there exists $k\in K(j)$, such that $k\in N(i)$. Since $j\in \overline{N_>(i)}$ we have that $j\neq k$, which implies that $k>j$. We then have vertices $i<j<k$ such that $k\in N(i)$, $j\notin N(i)$ and $k\notin N(j)$, which is a contradiction to the fact that $<$ is a DUF-ordering. Therefore we can conclude that $K(i)=\{i\}\cup K(j)$ is an independent set in $G[i,n]$. Since $j\in P_i$, we have by the definition of $P_i$ that $[i+1,j-1]\subseteq N^-(i)\cup N^-(j)$. It then follows from the fact that $K(j)$ is an absorbing set of $G[j,n]$ containing $j$ that $K(i)=\{i\}\cup K(j)$ is an absorbing set of $G[i,n]$. Thus $K(i)$ is a kernel of $G[i,n]$ that contains $i$.\medskip

\noindent\ref{kiy}  Suppose that $OPT(i)\neq\infty$. The proof is again by reverse induction on $i$. The base case where $i=n$ is trivially true. Assume that the hypothesis is true for any $j>i$. If $|OPT(i)|=1$, then it should be the case that $OPT(i)=\{i\}$ and $j\in N^-(i)$ for each $j\in \{i+1,\ldots,n\}$, i.e. $N^-_>(i)=\{i+1,\ldots,n\}$. By the definition of $K(i)$, we then have $K(i)= \{i\}$, and we are done. Therefore we can assume that $|OPT(i)|>1$. Let $j=\min (OPT(i)\setminus\{i\})$. Clearly, $j>i$. As $OPT(i)$ is an independent set, we have that $j\in \overline{N_>(i)}$. We claim that $j\in P_i$. Suppose that there exists a vertex $y\in [i+1,j-1]$ such that $y\notin N^-(i)\cup N^-(j)$. Since $OPT(i)$ is an absorbing set in $G[i,n]$, there exists a vertex $k\in OPT(i)\setminus\{i,j\}$ such that $y\in N^-(k)$. By the choice of $j$ and the definition of $k$, we have that $j<k$ and $(j,k)\notin E(G)$. Then we have $y<j<k$, $(y,k)\in E(G)$, and $(y,j),(j,k)\notin E(G)$, which is a contradiction to the fact that $<$ is a DUF-ordering. Therefore we can conclude that $[i+1,j-1]\subseteq N^-(i)\cup N^-(j)$, which implies by the definition of $P_i$ that $j\in P_i$. This proves our claim. Note that if there exists a vertex $z\in (N^-(i)\setminus N^-(j))\cap [j,n]$, then we have vertices $i<j<z$ such that $(z,i)\in E(G)$ and $(z,j),(j,i)\notin E(G)$, which is a contradiction to the fact that $<$ is a DUF-ordering. Therefore we can assume that $N^-(i)\cap [j,n]\subseteq N^-(j)\cap [j,n]$. This implies that $OPT(i)\setminus\{i\}$ is a kernel of $G[j,n]$ that contains $j$. Thus $OPT(j)\neq\infty$, which implies by the induction hypothesis that $K(j)\neq\infty$ and $|K(j)|=|OPT(j)|\leq|OPT(i)\setminus \{i\}|$. Since $j\in P_i$ and $K(j)\neq\infty$, we have $K(i)\neq\infty$, and further we have $|K(i)|\leq |\{i\}\cup K(j)|\leq 1+|OPT(i)\setminus \{i\}| = |OPT(i)|$. By~\ref{kernel}, $K(i)$ is a kernel of $G[i,n]$ that contains $i$, and hence we have $|K(i)|=|OPT(i)|$.
\end{proof}
Suppose that $G$ has a kernel. Now let $OPT$ denote a minimum sized kernel in $G$. Let $\mathcal{K}=\{K(j)\neq \infty:[1,j-1]\subseteq N^-(j)\}$. Note that we consider $[1,j-1]=\emptyset$ if $j=1$. By Lemma~\ref{lem:kernel}\ref{kernel}, it follows that every member of $\mathcal{K}$ is a kernel of $G$. So if $G$ does not have a kernel, then $\mathcal{K}=\emptyset$. The following lemma shows that the converse is also true.
\begin{lemma}\label{lem:OPTkernel}
If $G$ has a kernel, then $\mathcal{K}\neq\emptyset$ and $|OPT|= \big|\mathrm{Min}(\mathcal{K})\big|$. 
\end{lemma}
\begin{proof}
Suppose that $G$ has a kernel. Then clearly, $OPT$ exists.
Let $j=\min\{i: i\in OPT\}$. Then it should be the case that $[1,j-1]\subseteq N^-(j)$. As otherwise, there exist vertices $j'\in [1,j-1]$ and $k\in OPT$ such that $j'\in N^-(k)\setminus N^-(j)$. Since $OPT$ is an independent set, this implies  that we have vertices, $j'<j<k$ such that $(j',k)\in E(G)$ and $(j',j),(j,k)\notin E(G)$ which is a contradiction to the fact that $<$ is a DUF-ordering. Also by the choice of $j$, we have that $OPT\subseteq [j,n]$. Then as $OPT$ is a kernel of $G$, $OPT$ is a kernel of $G[j,n]$ that contains $j$. This implies that $OPT(j)\neq \infty$ and $|OPT(j)|\leq |OPT|$. Therefore by Lemma~\ref{lem:kernel}, we have that $K(j)\neq \infty$ and $|K(j)|=|OPT(j)|$. Thus $K(j)\in \mathcal{K}$, which implies that $\mathcal{K}\neq\emptyset$. Further, $\big|\mathrm{Min}(\mathcal{K})\big|\leq|K(j)|=|OPT(j)|\leq |OPT|$. Since every member of $\mathcal{K}$ is a kernel of $G$, it now follows that $\big|\mathrm{Min}(\mathcal{K})\big|=|OPT|$.
\end{proof}
We thus have the following theorem.
\begin{theorem}\label{thm:kernelduf}
The DUF-digraph $G$ has a kernel if and only if $K(j)\neq \infty$ for some $j$ such that $[1,j-1]\subseteq N^-(j)$. Further, if $G$ has a kernel, then the set $\{K(j)\neq\infty:[1,j-1]\subseteq N^-(j)\}$ contains a kernel of $G$ of minimum possible size.
\end{theorem}

Let $G$ be a DUF-digraph with vertex set $[1,n]$. For each $i\in [1,n]$, we can compute the set $P_i$ in $O(n+m)$ time as follows. We mark the in-neighbors of $i$ in $[i,n]$ and then scan the vertices from $i$ to $n$ in a single pass in order to collect the vertices which are not in-neighbors of $i$ in an ordered list $L$. Initialize $P_i=\emptyset$. We mark every out-neighbor of $i$ in $L$. Now for each unmarked vertex $j$ in $L$ (processed from left to right), we add $j$ to $P_i$ if and only if every vertex of $L$ before $j$ is an in-neighbor of $j$. Note that this computation of $P_i$ can be done in $O(n+m)$ time. This implies that we can precompute the set $\{P_i:i\in [1,n]\}$ in $O((n+m)n)$ time. Now since $|P_i|\leq n$, it is easy to see from the recursive definition for $K(i)$ that $\{K(i):i\in [1,n]\}$ can be computed in $O(n^2)$ time. For $j\in [1,n]$, we can check in $O(n+m)$ time whether $[1,j-1]\subseteq N^-(j)$. Thus in $O((n+m)n)$ time, we can compute the minimum sized set in $\{K(j)\neq\infty:[1,j-1]\subseteq N^-(j)\}$. 
Therefore by Theorem~\ref{thm:kernelduf}, we have the following corollary.  

\begin{corollary}\label{cor:algdufker}
The {\sc Min-Kernel} problem can be solved for DUF-digraphs in $O((n+m)n)$ time if the DUF-ordering is known. Consequently, for a reflexive interval digraph, the {\sc Min-Kernel} problem can be solved in $O((n+m)n)$ time if the interval representation is given as input.
\end{corollary}

Let $G$ be a cocomparability graph. Let $H$ be the symmetric digraph of $G$. Now it is easy to see that a set $K\subseteq V(H)=V(G)$, is a kernel of $H$ if and only if $K$ is an independent dominating set of $G$. Therefore a kernel of minimum possible size in $H$ will be a minimum independent dominating set in $G$. Note that a vertex ordering of a cocomparability graph that satisfies the properties in Theorem~\ref{thm:cocomp} can be found in linear time~\cite{mcconnell1999modular}. Let $<$ be such a vertex ordering of $G$. As noted before, $H$ is a DUF-digraph with DUF-ordering $<$. Thus an algorithm that computes a minimum sized kernel in $H$ also computes a minimum independent dominating set in $G$. From Corollary~\ref{cor:algdufker}, we now have the following.
\begin{corollary}
An independent dominating set of minimum possible size can be found in $O((n+m)n)$ time in cocomparability graphs.
\end{corollary}
The above corollary is an improvement over the result by Kratsch and Stewart~\cite{kratsch1993domination} that an
independent dominating set of minimum possible size problem can be computed in $O(n^3)$ time for cocomparability graphs.
\smallskip

We now show that a minimum sized kernel of an adjusted interval digraph, whose interval representation is known, can be computed more efficiently than in the case of DUF-digraphs.
Let $G$ be an adjusted interval digraph with an interval representation $\{(S_u,T_u)\}_{u\in V(G)}$. Note that by the definition of adjusted interval digraphs, we have that $l(S_u)=l(T_u)$ for each $u\in V(G)$. Let $<$ be an ordering of vertices in $G$ with respect to the common left end points of intervals corresponding to each vertex. Then $<$ has the following property: for any three distinct vertices $u<v<w$, if $(u,w)\in E(G)$ then $(u,v)\in E(G)$ and if $(w,u)\in E(G)$ then $(v,u)\in E(G)$. Then note that $<$ is also a DUF-ordering of $V(G)$. Further, for each vertex $v\in V(G)$, the vertices in $N^+_>(v)$ and $N^-_>(v)$ occur consecutively in $<$. This implies that for each vertex $v\in V(G)$ the vertices in $N_>(v)$ occur consecutively in $<$. Further we have
\begin{equation}\label{eq:adjprop}
\mbox{if }[x,y]\subseteq N^-(y)\mbox{ (resp. }N^+(y)\mbox{) then for any }z\in [x,y]\mbox{, we have }[x,z]\subseteq N^-(z)\mbox{ (resp. }N^+(z)\mbox{).}
\end{equation}

Let $V(G)=[1,n]$ and $<$ be the ordering $(1,2,\ldots,n)$.
We can compute the sets $\{\max N^+(i)\colon i\in [1,n]\}$ and $\{\max N^-(i)\colon i\in [1,n]\}$ in $O(n+m)$ time by just preprocessing the adjacency list of $G$. Since the vertices in $N^+_>(i)$ (resp. $N^-_>(i)$) occur consecutively in $<$, we can also compute the set $\{\min \overline{N_>(i)}=\max N(i)+1\colon i\in [1,n]\}$ in $O(n+m)$ time (if $\max N(i)=n$, then we set $\min \overline{N_>(i)}=n+1$). Let $i\in [1,n]$. We can construct $P_i$ as follows. We compute $x=\min\{\max N^+(j)\colon j\in [\min\overline{N^-_>(i)}=\max N^-(i)+1,n]\}$ in $O(n)$ time (note that if $\max N^-(i)=n$, then $\min\overline{N_>(i)}=n+1$, in which case we can just set $P_i=\emptyset$). We claim that $P_i=[\min \overline{N_>(i)},x]$. To see this, first note that for every vertex $i$, $N^+_>(i)=[i, \max N^+(i)]$. Therefore, since for any vertex $z\in [\min \overline{N_>^-(i)},x]$, we have $\max N^+(z)\geq x$, we can conclude that $(z,x)\in E(G)$. Thus $[\min \overline{N_>^-(i)},x]\subseteq N^-(x)$. Therefore by property~\eqref{eq:adjprop}, for each $z\in [\min \overline{N_>(i)},x]\subseteq [\min \overline{N_>^-(i)},x]$, we have that $[\min \overline{N_>^-(i)},z]\subseteq N^-(z)$. Since the in-neighbors of $i$ are consecutive in $<$, this means that $[i,z]\subseteq N^-(i)\cup N^-(z)$. Therefore we have that $[\min \overline{N_>(i)},x]\subseteq P_i$. Now consider any $z>x$. By the definition of $x$, there exists $j\in [\min \overline{N_>^-(i)},x]$ such that $\max N^+(j)=x$. Then as $x<z$, we have $(j,z)\notin E(G)$, which implies that $j\notin N^-(i)\cup N^-(z)$. Thus $z\notin P_i$. Therefore we can conclude that $P_i=[\min \overline{N_>(i)},x]$. Note that if $\min \overline{N_>(i)}>x$, then $P_i=\emptyset$. It is clear that the set $P_i$ can be computed in this way in $O(n)$ time for an $i\in [1,n]$. So the set $\{P_i:i\in [1,n]\}$ can be computed in $O(n^2)$ time. The sets $\{K(i)\colon i\in [1,n]\}$ can then be computed in $O(n^2)$ time as before. Now we compute $y=\min\{\max N^+(j)\colon j\in [1,n]\}$ in $O(n)$ time. Then $[1,y-1]\subseteq N^-(y)$. Therefore by property~\eqref{eq:adjprop}, for each $z\in [1,y]$ we have $[1,z-1]\subseteq N^-(z)$. Now consider any $z>y$. By the definition of $y$, there exists $j\in [1,y]$ such that $y=\max N^+(j)$. Then as $y<z$, we have that $(j,z)\notin E(G)$, which implies that $[1,z-1]\not\subseteq N^-(z)$. Therefore we can conclude that $[1,y]=\{j\colon [1,j-1]\subseteq N^-(j)\}$. By Theorem~\ref{thm:kernelduf}, we can just output in $O(n)$ time a set of minimum size in $\{K(i)\colon i\in [1,y]$ and $K(i)\neq\infty\}$ as a minimum sized kernel of $G$. Thus we have the following corollary.
\begin{corollary}
The {\sc Min-Kernel} problem can be solved in $O(n^2)$ time in adjusted interval digraphs, given an adjusted interval representation of the input graph.
\end{corollary}
\begin{remark}
Note that the {\sc Max-Kernel} problem can also be solved in $O((n+m)n)$ time for the class of DUF-digraphs, by a minor modification of our algorithm that solves {\sc Min-Kernel} problem (replace $\mathrm{Min}\{K(j)\neq\infty:j\in P_i \}$ in the recursive definition of $K(i)$ by $\mathrm{Max}\{K(j)\neq\infty:j\in P_i \}$ and follow the same procedure. Then we have that if kernel exists, then a maximum sized kernel is given by $\mathrm{Max} (\mathcal{K})$). Further, the recursive definition can also be easily adapted to the weighted version of the problems {\sc Min-Kernel} and {\sc Max-Kernel} in $O((n+m)n)$ time.  
\end{remark}

\subsection{Minimum absorbing set}
Given any digraph $G$, the splitting bigraph $B_G$ is defined as follows: $V(B_G)$ is partitioned into two sets $V'=\{u':u\in V(G)\}$ and $V''=\{u''\colon u\in V(G)\}$, and $E(B_G)=\{u'v''\colon (u,v)\in E(G)\}$. Muller~\cite{muller1997recognizing} observed that $G$ is an interval digraph if and only if $B_G$ is an interval bigraph (since if $\{(S_u,T_u)\}_{u\in V(G)}$ is an interval representation of a digraph $G$, then $\{\{S_u\}_{u'\in V'},\{T_u\}_{u''\in V''}\}$ is an interval bigraph representation of the bipartite graph $B_G$).

Recall that for a bipartite graph having two specified partite sets $A$ and $B$, a set $S\subseteq B$ such that $\bigcup_{u\in B} N(u)=A$ is called an $A$-dominating set (or a red-blue dominating set). If $G$ is a reflexive interval digraph, then every $V'$-dominating set of $B_G$ corresponds to an absorbing set of $G$ and vice versa. To be precise, if $S\subseteq V''$ is a $V'$-dominating set of $B_G$, then $\{u\colon u''\in S\}$ is an absorbing set of $G$ and if $S\subseteq V(G)$ is an absorbing set of $G$, then $\{u''\colon u\in S\}$ is a $V'$-dominating set of $B_G$ (note that this is not true for general interval digraphs). Thus finding a minimum cardinality absorbing set in $G$ is equivalent to finding a minimum cardinality $V'$-dominating set in the bipartite graph $B_G$. We show in this section that the problem of computing a minimum cardinality $A$-dominating set is linear time solvable for interval bigraphs. 
This implies that the {\sc Absorbing-Set} problem can be solved in linear time on reflexive interval digraphs.

Consider an interval bigraph $H$ with partite sets $A$ and $B$. Let $\{I_u\}_{u\in V(H)}$ be an interval representation for $H$; i.e. $uv\in E(H)$ if and only if $u\in A$, $v\in B$ and $I_u\cap I_v\neq\emptyset$. Let $|A|=t$. We assume without loss of generality that $A=\{1,2,\ldots,t\}$, where $r(I_i)<r(I_j)\Leftrightarrow i<j$. We also assume that there are no isolated vertices in $A$, as otherwise $H$ does not have any $A$-dominating set. For each $i\in\{1,2,\ldots,t\}$, we compute a minimum cardinality subset $DS(i)$ of $B$ that dominates $\{i,i+1,\ldots,t\}$, i.e. $\{i,i+1,\ldots,t\}\subseteq\bigcup_{u\in DS(i)} N(u)$. Then $DS(1)$ will be a minimum cardinality $A$-dominating set of $H$. We first define some parameters that will be used to define $DS(i)$.

Let $i\in\{1,2,\ldots,t\}$. We define $\rho(i)=\max_{u\in N(i)} r(I_u)$ and let $R(i)$ be a vertex in $N(i)$ such that $r(I_{R(i)})=\rho(i)$. Since $A$ does not contain any isolated vertices, $\rho(i)$ and $R(i)$ exist for each $i\in\{1,2,\ldots,t\}$. Let $\lambda(i)=\min\{j\colon\rho(i)<l(I_j)\}$. Note that $\lambda(i)$ may not exist. It can be seen that if $\lambda(i)$ exists, then $\lambda(i)>i$ in the following way. Let $j=\lambda(i)$. Clearly, $\rho(i)<l(I_j)$. As $R(i)\in N(i)$, we have $l(I_i)<\rho(i)$, which implies that $i\neq j$. If $j<i$, then it should be the case that $l(I_i)<\rho(i)<l(I_j)<r(I_j)<r(I_i)$, which implies that any interval $I_x$, where $x\in B$, that intersects $I_j$ also intersects $I_i$, and $r(I_x)>\rho(i)$. But this contradicts our choice of $\rho(i)$ and $R(i)$. Thus $N(j)=\emptyset$, implying that $j$ is an isolated vertex in $A$, which is a contradiction. Therefore, we can conclude that for any $i\in A$, $\lambda(i)>i$.
\begin{lemma}\label{lem:lambda}
Let $i\in\{1,2,\ldots,t\}$. If $\lambda(i)$ exists, then $R(i)$ dominates every vertex in $\{i,i+1,\ldots,\lambda(i)-1\}$ and otherwise, $R(i)$ dominates every vertex in $\{i,i+1,\ldots,t\}$.
\end{lemma}
\begin{proof}
We first note that as $R(i)\in N(i)$, we have $l(I_{R(i)})\leq r(I_i)$, as otherwise the intervals $I_{R(i)}$ and $I_i$ will be disjoint.

Suppose that $\lambda(i)$ exists. Then consider any $j\in\{i,i+1,\ldots,\lambda(i)-1\}$. Suppose for the sake of contradiction that $R(i)\notin N(j)$. Clearly, $j\neq i$ as $R(i)\in N(i)$. So we have $i< j< \lambda(i)$. Since $I_{R(i)}$ and $I_j$ are disjoint, we have either $\rho(i)=r(I_{R(i)})<l(I_j)$ or $r(I_j)<l(I_{R(i)})$. In the former case, since $i<j<\lambda(i)$, we have a contradiction to the choice of $\lambda(i)$. So we can assume that $r(I_j)<l(I_{R(i)})$. Recalling that $l(I_{R(i)})\leq r(I_i)$, we now have that $r(I_j)<r(I_i)$, which contradicts the fact that $j>i$. Thus, $R(i)$ dominates every vertex in $\{i,i+1,\ldots,\lambda(i)-1\}$. Next, suppose that $\lambda(i)$ does not exist. Then consider any vertex $j>i$. Since $\lambda(i)$ does not exist, we have $l(I_j)\leq\rho(i)=r(I_{R(i)})$. Since $l(I_{R(i)})\leq r(I_i)$ and $r(I_i)<r(I_j)$, we have $l(I_{R(i)})<r(I_j)$. Thus, the intervals $I_j$ and $I_{R(i)}$ intersect for every $j>i$, implying that $R(i)$ dominates every vertex in $\{i,i+1,\ldots,t\}$.
\end{proof}

We now explain how to compute $DS(i)$ for each $i\in\{1,2,\ldots,t\}$.
We recursively define $DS(i)$ as follows:
\medskip

\noindent $DS(i) =  \begin{cases} \{R(i)\}\cup DS(\lambda(i)) & \text{ if } \lambda(i)\text{ exists}
 \\ \{R(i)\} & \text{ otherwise}
\end{cases}$

\begin{lemma}
For each $i\in \{1,2,\ldots,t\}$, the set $DS(i)$ as defined above is a minimum cardinality subset of $B$ that dominates $\{i,i+1,\ldots,t\}$.
\end{lemma}
\begin{proof}
We prove this by induction on $t-i$. The base case where $i=t$ is trivial, by the definition of $R(t)$. Let $i<t$. Assume that the hypothesis holds for any $j>i$. If $\lambda(i)$ does not exist, then by Lemma~\ref{lem:lambda}, $R(i)$ dominates every vertex in $\{i,i+1,\ldots,t\}$. This implies that $DS(i)=\{R(i)\}$ is a minimum cardinality subset of $B$ that dominates $\{i,i+1,\ldots,t\}$ and we are done. Therefore let us assume that $\lambda(i)$ exists. Then by the recursive definition of $DS(i)$, we have that $DS(i)=\{R(i)\} \cup DS(\lambda(i))$. Since $\lambda(i)> i$, we have by the inductive hypothesis that $DS(\lambda(i))$ is a minimum cardinality subset of $B$ that dominates every vertex in $\{\lambda(i),\lambda(i)+1,\ldots,t\}$. Since by Lemma~\ref{lem:lambda}, we have that $R(i)$ dominates every vertex in $\{i,i+1,\ldots,\lambda(i)-1\}$, we then have that $DS(i)=\{R(i)\}\cup DS(\lambda(i))$ dominates every vertex in $\{i,i+1,\ldots,t\}$. Consider any set $OPT\subseteq B$ that dominates $\{i,i+1,\ldots,t\}$. Clearly, there exists a $u\in OPT$ such that $i\in N(u)$. By the definition of $R(i)$, we know that $r(I_u)\leq r(I_{R(i)})=\rho(i)$. Since $\rho(i)<l(I_{\lambda(i)})$, this implies that $\lambda(i)\notin N(u)$. Then, since $\lambda(i)\in\{i,i+1,\ldots,t\}$, there must exist a vertex $v\in OPT\setminus\{u\}$ such that $\lambda(i)\in N(v)$. Now consider any vertex $j\in N(u)\cap\{\lambda(i),\lambda(i)+1,\ldots,t\}$. We have $r(I_j)\geq r(I_{\lambda(i)})\geq l(I_{\lambda(i)})>\rho(i)\geq r(I_u)$. Since $j\in N(u)$, we have $l(I_j)\leq r(I_u)$, which implies that $l(I_j)<l(I_{\lambda(i)})\leq r(I_{\lambda(i)})<r(I_j)$. This implies that every interval that intersects $I_{\lambda(i)}$ also intersects $I_j$, in particular $j\in N(v)$. Applying the argument for every $j\in N(u)\cap\{\lambda(i),\lambda(i)+1,\ldots,t\}$, we can conclude $N(u)\cap\{\lambda(i),\lambda(i)+1,\ldots,t\}\subseteq N(v)$. Since $OPT$ dominates every vertex in $\{\lambda(i),\lambda(i)+1,\ldots,t\}$, this implies that $OPT\setminus\{u\}$ dominates every vertex in $\{\lambda(i),\lambda(i)+1,\ldots,t\}$.
Since by the inductive hypothesis, $DS(\lambda(i))$ is a minimum cardinality subset of $B$ that dominates every vertex in the same set, we have that $|OPT\setminus\{u\}|\geq |DS(\lambda(i))|$. Then $|OPT|\geq |DS(\lambda(i))\cup\{R(i)\}|=|DS(i)|$. This proves that $DS(i)$ is a minimum cardinality subset of $B$ that dominates every vertex in $\{i,i+1,\ldots,t\}$.
\end{proof}

It is not difficult to verify that given an interval representation of the interval bigraph $H$ with partite sets $A$ and $B$, the parameters $R(i)$ and $\lambda(i)$ can be computed for each $i\in A$ in $O(n+m)$ time. Also, given a reflexive interval digraph $G$, the interval bigraph $B_G$ can be constructed in linear time. Thus we have the following corollary.
\begin{corollary}
The {\sc Red-Blue Dominating Set} problem can be solved in interval bigraphs in linear time, given an interval representation of the bigraph as input. Consequently, the {\sc Absorbing-Set} (resp. {\sc Dominating-Set})problem can be solved in linear time in reflexive interval digraphs, given an interval representation of the input digraph. 
\end{corollary}

Note that even if an interval representation of the interval bigraph is not known, it can be computed in polynomial time using Muller's algorithm~\cite{muller1997recognizing}. Thus given just the adjacency list of the graph as input, the {\sc Red-Blue Dominating Set} problem is polynomial-time solvable on interval bigraphs and the {\sc Absorbing-Set} (resp. {\sc Dominating-Set}) problem is polynomial-time solvable on reflexive interval digraphs.

\subsection{Maximum independent set}
We have the following theorem due to McConnell and Spinrad~\cite{mcconnell1999modular}.
\begin{theorem}
An independent set of maximum possible size can be computed for cocomparability graphs in $O(n+m)$ time.
\end{theorem}
Let $G$ be a DUF-digraph. Let $H$ be the underlying undirected graph of $G$. Then by Corollary~\ref{cor:DUFcc}, we have that $H$ is a co-comparability graph. Note that the independent sets of $G$ and $H$ are exactly the same. Therefore any algorithm that finds a maximum cardinality independent set in cocomparability graphs can be used to solve the {\sc Independent-Set} problem in DUF-digraphs. Thus by the above theorem, we have the following corollary.
\begin{corollary}
The {\sc Independent-Set} problem can be solved for DUF-digraphs in $O(n+m)$ time. Consequently, the {\sc Independent-Set} problem can be solved for reflexive interval digraphs in $O(n+m)$ time. 
\end{corollary}
The above corollary generalizes and improves the $O(mn)$ time algorithm due to Prisner's~\cite{prisner1994algorithms} observation that underlying undirected graph of interval nest digraphs are weakly chordal (Theorem~\ref{thm:prisnerweaklytriang}) and the fact that maximum cardinality independent set problem can be solved for weakly chordal graphs in $O(mn)$ time~\cite{hss}. Note that the weighted {\sc Independent-Set} problem  can also be solved for DUF-digraphs in $O(n+m)$ time, as the problem of finding a maximum weighted independent set in a cocomparability graphs can be solved in linear time~\cite{kohler2016linear}. 
\section{Hardness results for point-point digraphs}\label{sec:npcomplete}
\subsection{Characterizations for point-point digraphs}
In this section we give a characterization for point-point digraphs which will be further useful for proving our NP-completeness results for this class. Let $G=(V,E)$ be a digraph. We say that $a,b,c,d$ is an \emph{anti-directed walk} of length 3 if $a,b,c,d\in V(G)$, $(a,b),(c,b),(c,d)\in E(G)$ and $(a,d)\notin E(G)$ (the vertices $a,b,c,d$ need not be pairwise distinct, but it follows from the definition that $a\neq c$ and $b\neq d$). Recall that $B_G=(X,Y,E)$ is a splitting bigraph of $G$, where $X=\{x_u:u\in V(G)\}$ and $Y=\{y_u:u\in V(G)\}$ and $x_uy_v\in E(G_B)$ if and only if $(u,v)\in E(G)$. We then have the following theorem. 
\begin{theorem} \label{thm:point-point}
Let $G$ be a digraph. Then the following conditions are equivalent:
\begin{myenumerate}
	\item \label{point} $G$ is a point-point digraph.
	\item \label{anti-walk} $G$ does not contain any anti-directed walk of length 3.
	\item \label{complete} The splitting bigraph of $G$ is a disjoint union of complete bipartite graphs.  	
\end{myenumerate}
\end{theorem}
\begin{proof}
$\ref{point}  \Rightarrow \ref{anti-walk}$: Let $G$ be a point-point digraph with a point-point representation $\{(S_u,T_u)\}_{u\in V(G)}$. Suppose that there exist vertices $a,b,c,d$ in $G$ such that $(a,b),(c,b),(c,d)\in E(G)$. By the definition of point-point representation, we then have $S_a=T_b=S_c=T_d$. This implies that $(a,d)\in E(G)$. Therefore we can conclude that $G$ does not contain any anti-directed walk of length 3.

\smallskip
\noindent $\ref{anti-walk}  \Rightarrow \ref{complete}$: Suppose that $G$ does not contain any anti-directed walk of length 3. Let $B_G=(X,Y,E)$ be the splitting bigraph of $G$. Let $x_uy_v$ be any edge in $B_G$, where $u,v\in V(G)$. Clearly, by the definition of $B_G$, $(u,v)\in E(G)$. We claim that the graph induced in $B_G$ by the vertices $N(x_u)\cup N(y_v)$ is a complete bipartite graph. Suppose not. Then it should be the case that there exist two vertices $x_a\in N(y_v)$ and $y_b\in N(x_u)$ such that $x_ay_b\notin E(B_G)$, where $a,b\in V(G)$. By the definition of $B_G$, we then have that $(a,v),(u,v),(u,b)\in E(G)$ and $(a,b)\notin E(G)$. So $a,v,u,b$ is an anti-directed walk of length 3 in $G$, which is a contradiction to \ref{anti-walk}. This proves that for every $p\in X$ and $q\in Y$ such that $pq\in E(B_G)$, the set $N(p)\cup N(q)$ induces a complete bipartite subgraph in $B_G$. Therefore, each connected component of $B_G$ is a complete bipartite graph. (This can be seen as follows: Suppose that there is a connected component $C$ of $B_G$ that is not complete bipartite. Choose $p\in X\cap C$ and $q\in Y\cap C$ such that $pq\notin E(B_G)$ and the distance between $p$ and $q$ in $B_G$ is as small as possible. Let $t$ be the distance between $p$ and $q$ in $B_G$. Clearly, $t$ is odd and $t\geq 3$. Consider a shortest path $p=z_0,z_1,z_2,\ldots,z_t=q$ from $p$ to $q$ in $B_G$. By our choice of $p$ and $q$, we have that $z_1z_{t-1}\in E(B_G)$. But then $p\in N(z_1)$, $q\in N(z_{t-1})$ and $pq\notin E(B_G)$, contradicting our observation that $N(z_1)\cup N(z_{t-1})$ induces a complete bipartite graph in $B_G$.)
\smallskip

\noindent $\ref{complete} \Rightarrow \ref{point}$: Suppose that $G$ is a digraph such that the splitting bigraph $B_G$ is a disjoint union of complete bipartite graphs, say $H_1,H_2,\ldots,H_k$. Now we can obtain a point-point representation $\{(S_u,T_u)\}_{u\in V(G)}$ of the digraph $G$ as follows: For each $i\in \{1,2,\ldots,k\}$, define $S_u=i$ if $x_u\in V(H_i)$ and $T_v=i$ if $y_v\in V(H_i)$. Note that $(u,v)\in E(G)$ if and only if $x_uy_v\in E(B_G)$ if and only if $x_u,y_v\in V(H_i)$ for some $i\in \{1,2,\ldots,k\}$. Therefore we can conclude that $(u,v)\in E(G)$ if and only if $S_u=T_v=i$ for some $i\in \{1,2,\ldots,k\}$. Thus the digraph $G$ is a point-point digraph.
\end{proof}
\begin{corollary}
Point-point digraphs can be recognized in linear time.
\end{corollary}
\subsection{Subdivision of an irreflexive digraph}\label{sec:subdivision}
For an undirected graph $G$, the \emph{$k$-subdivision} of $G$, where $k\geq 1$, is defined as the graph $H$ having vertex set $V(H)=V(G)\cup\bigcup_{ij\in E(G)} \{u_{ij}^1,u_{ij}^2,\ldots,u_{ij}^k\}$, obtained from $G$ by replacing each edge $ij\in E(G)$ by a path $i,u_{ij}^1,u_{ij}^2,\ldots,u_{ij}^k,j$.
 
The following theorem is adapted from Theorem~5 of Chleb\'ik and Chleb\'ikov\'a~\cite{chlebik2007complexity}.
\begin{theorem}[Chleb\'ik and Chleb\'ikov\'a]\label{thm:indsetsub}
Let $G$ be an undirected graph having $m$ edges. Let $k\geq 1$.
\begin{myenumerate}
	\item \label{indpt} The problem of computing a maximum cardinality independent set is APX-complete when restricted to $2k$-subdivisions of 3-regular graphs for any fixed integer $k\geq 0$.
	\item The problem of finding a minimum cardinality dominating set (resp. independent dominating set) is APX-complete when restricted to $3k$-subdivisions of graphs having degree at most 3 for any fixed integer $k\geq 0$.
\end{myenumerate}
\end{theorem}
Note that the independent sets, dominating sets and independent dominating sets of an undirected graph $G$ are exactly the independent sets, dominating sets (which are also the absorbing sets), and solutions (which are also the kernels) of the symmetric digraph of $G$. Clearly the symmetric digraph of $G$ is irreflexive. Since the {\sc Max-Kernel} problem is equivalent to the {\sc Independent-Set} problem in symmetric digraphs, we then have the following corollary of Theorem~\ref{thm:indsetsub}.
\begin{corollary}\label{cor:APX}
	The problems {\sc Independent-Set}, {\sc Absorbing-Set}, {\sc Min-Kernel} and {\sc Max-Kernel} problems are APX-complete on irreflexive symmetric digraphs of in- and out-degree at most 3.  
\end{corollary} 

Suppose that $k\geq 0$. Let $H$ be the $2k$-subdivision or $3k$-subdivision of an undirected graph and let $G$ be the symmetric digraph of $H$. Note that the independent sets, dominating sets, and independent dominating sets of $H$ are exactly the independent sets, dominating sets (which are also the absorbing sets), and solutions (which are also the kernels) of $G$. Therefore from Theorem~\ref{thm:indsetsub} we have that the {\sc Independent-Set} problem is APX-hard on irreflexive symmetric digraphs of $2k$-subdivisions of 3-regular graphs, and that the {\sc Absorbing-Set} and {\sc Min-Kernel} problems are APX-hard on the symmetric digraphs of $3k$-subdivisions of graphs of degree at most 3 for each $k\geq 0$. 
But note that for $k\geq 1$, the symmetric digraph of the $2k$-subdivision or $3k$-subdivision of an undirected graph contains an anti-directed walk of length 3 (unless the graph contains no edges), and therefore by Theorem~\ref{thm:point-point}, is not a point-point digraph. Thus we cannot directly deduce the APX-hardness of the problems under consideration for point-point digraphs from Theorem~\ref{thm:indsetsub}.

We define the subdivision of an irreflexive digraph, so that the techniques of Chleb\'ik and Chleb\'ikov\'a can be adapted for proving hardness results on point-point digraphs.

\begin{definition}
Let $G$ be an irreflexive digraph (i.e. $G$ contains no loops). For $k\geq 1$, define the \emph{$k$-subdivision} of $G$ to be the digraph $H$ having vertex set $V(H)=V(G)\cup\bigcup_{(i,j)\in E(G)} \{u_{ij}^1,u_{ij}^2,\ldots,u_{ij}^k\}$, obtained from $G$ by replacing each edge $(i,j)\in E(G)$ by a directed path $i,u_{ij}^1,u_{ij}^2,\ldots,u_{ij}^k,j$.
\end{definition}
Note that the $k$-subdivision of any irreflexive digraph is also an irreflexive digraph. We then have the following lemma.
\begin{lemma} \label{lem:subdiv_point}
For any $k\geq 1$, the $k$-subdivision of any irreflexive digraph is a point-point digraph.
\end{lemma}
\begin{proof}
Let $k\geq 1$ and let $G$ be any irreflexive digraph.
By Theorem~\ref{thm:point-point}, it is enough to show that the $k$-subdivision $H$ of $G$ does not contain any anti-directed walk of length 3. Note that by the definition of $k$-subdivision, all the vertices in $V(H)\setminus V(G)$ have both in-degree and out-degree exactly equal to one. Further, for every vertex $v$ in $H$ such that $v\in V(G)$, we have that $N^+(v),N^-(v)\subseteq V(H)\setminus V(G)$. Suppose for the sake of contradiction that $u,v,w,x$ is an anti-directed walk of length 3 in $H$. Recall that we then have $(u,v),(w,v),(w,x)\in E(H)$, $u\neq w$ and $v\neq x$. By the above observations, we can then conclude that $v\in V(G)$ and further that $u,w\in V(H)\setminus V(G)$. Then since $(w,x)\in E(H)$ and $v\neq x$, we have that $w$ has out-degree at least 2, which contradicts our earlier observation that every vertex in $V(H)\setminus V(G)$ has out-degree exactly one. This proves the lemma.
\end{proof}

\begin{theorem}
The problem {\sc Independent-Set} is APX-hard for point-point digraphs.
\end{theorem}
\begin{proof}
We show a reduction from the {\sc Independent Set} problem in undirected graphs. Let $G$ be an undirected graph. Let $D$ be the digraph obtained by assigning an arbitrary direction to each edge of $G$. Clearly, $D$ is irreflexive. Let $D'$ be a 2-subdivision of the directed graph $D$. Note that the underlying undirected graph $H$ of $D'$ is a 2-subdivision of $G$. Thus, given the 2-subdivision of an undirected graph $G$, we can construct in polynomial time the graph $D'$, and moreover, by Lemma~\ref{lem:subdiv_point}, $D'$ is a point-point digraph. Since the independent sets of the 2-subdivision of $G$ are exactly the independent sets of $D'$, we can conclude from Theorem~\ref{thm:indsetsub}\ref{indpt} that the problem {\sc Independent-Set} is APX-hard for point-point digraphs.
\end{proof}
\subsection{Kernel}
\begin{lemma}\label{lem:kernel_npc}
Let $G$ be an irreflexive digraph and let $k\geq 1$. Then $G$ has a kernel if and only if the $2k$-subdivision of $G$ has a kernel. Moreover, $G$ has a kernel of size $q$ if and only if the $2k$-subdivision of $G$ has a kernel of size $q+km$. Further, given a kernel of size $q+km$ of the $2k$-subdivision of $G$, we can construct a kernel of size $q$ of $G$ in polynomial time.
\end{lemma}
\begin{proof}
	Let $H$ be the $2k$-subdivision of $G$ and let $\bigcup_{(i,j)\in E(G)} \{u_{ij}^1,u_{ij}^2,\ldots,u_{ij}^{2k}\}$ be the vertices in $V(H)\setminus V(G)$ as defined in Section~\ref{sec:subdivision}.
	
	Suppose that $G$ has a kernel $K\subseteq V(G)$. We define the set $K'\subseteq V(H)$ as $K'=K\cup\bigcup_{(i,j)\in E(G)} S(i,j)$, where
	$$S(i,j)=\begin{cases}
	\{u_{ij}^{2l}: l\in \{1,2,\ldots,k\}\}, & \text{ if } j\notin K \\
	\{u_{ij}^{2l-1}: l\in \{1,2,\ldots,k\}\}, & \text{ if } j\in K 
	\end{cases}$$  
	
	We claim that $K'$ is a kernel in $H$. Note that as $K$ is an independent set in $G$, for any edge $(i,j)\in E(G)$, we have that $i\notin K$ whenever $j\in K$. Thus by the definition of $2k$-subdivision and $K'$, it is easy to see that $K'$ is an independent set in $H$. Therefore in order to prove our claim, it is enough to show that $K'$ is an absorbing set in $H$. Consider any $(i,j)\in E(G)$. It is clear from the definition of $K'$ that for each $t\in\{1,2,\ldots,2k-1\}$, either the vertex $u_{ij}^t$ or $u_{ij}^{t+1}$ is in $K'$. Further, we also have that either the vertex $u_{ij}^{2k}$ or $j$ is in $K'$. Thus for every vertex $x\in V(H)\setminus V(G)$, either $x$ or one of its out-neighbours is in $K'$. Now consider a vertex $i$ in $V(H)$ such that $i\in V(G)$. If $i\in K$, then $i\in K'$. On the other hand if $i\notin K$, then since $K$ is a kernel of $G$, there exists an out-neighbour $j$ of $i$ such that $j\in K$, in which case we have $u_{ij}^1\in K'$. Thus in any case, either $i$ or an out-neighbour of $i$ is in $K'$. This shows that $K'$ is a kernel of $H$.	
	
	Note that by the definition of $K'$, we have $|K'\setminus K|=km$. Therefore if $|K|=q$ then $|K'|= q+km$.
	
	Now suppose that $K'\subseteq V(H)$ is a kernel in $H$.
	\begin{claim}\label{claim:path}
		Let $(i,j)\in E(G)$ and $t\in\{1,2,\ldots,2k-1\}$. Then $u_{ij}^t\in K'$ if and only if $u_{ij}^{t+1}\notin K'$.
	\end{claim}
	If $u_{ij}^t\in K'$, then since $K'$ is an independent set in $H$, we have $u_{ij}^{t+1}\notin K'$. On the other hand if $u_{ij}^t\notin K'$, then since $K'$ is an absorbing set in $H$, we have $u_{ij}^{t+1}\in K'$. 
	This proves the claim.\medskip
	
	We first show that $K'\cap V(G)$ is an independent set of $G$. Consider any edge $(i,j)\in E(G)$. Suppose that $i\in K'$. Then since $K'$ is an independent set in $H$, we have $u_{ij}^1\notin K'$. Applying Claim~\ref{claim:path} repeatedly, we have that $u_{ij}^{2k}\in K'$, which implies that $j\notin K'$. Thus, the set $K'\cap V(G)$ is an independent set in $G$. Next, we note that $K'\cap V(G)$ is also an absorbing set of $G$. To see this, consider any vertex $i$ of $H$ such that $i\in V(G)$. If $i\notin K'$, then since $K'$ is an absorbing set in $H$, there exists $(i,j)\in E(G)$ such that $u_{ij}^1\in K'$. Applying Claim~\ref{claim:path} repeatedly, we have that $u_{ij}^{2k}\notin K'$. Then since $K'$ is an absorbing set in $H$, we have that $j\in K'$. Thus $K'\cap V(G)$ is an absorbing set of $G$, which implies that $K'\cap V(G)$ is a kernel of $G$.
	
	Note that by Claim~\ref{claim:path}, we have that $|K' \setminus V(G)|\leq km$. Let $ij\in E(G)$. Since $K'$ is an absorbing set in $H$, for each $t\in \{1,2,\ldots,2k-1\}$, either $u_{ij}^t\in K'$ or $u_{ij}^{t+1}\in K'$. This implies that $|K'\cap \{u_{ij}^1,u_{ij}^2,\ldots,u_{ij}^{2k}\}|\geq k$. This further implies that $|K' \setminus V(G)|\geq km$. Therefore we can conclude that $|K' \setminus V(G)|= km$. Thus, if $|K'|= q+km$ then $|K'\cap V(G)|= q$. Clearly, given the kernel $K'$ of $H$, the kernel $K'\cap V(G)$ of $G$ can be constructed in polynomial time.	
\end{proof}

\begin{theorem}
The problem {\sc Kernel} is NP-complete for point-point digraphs.
\end{theorem}
\begin{proof}
We show a reduction from the {\sc Kernel} problem in general digraphs to the {\sc Kernel} problem in point-point digraphs. Let $G$ be any digraph. Let $G'$ be the digraph obtained from $G$ by removing all loops in it. Then note that the kernels in $G$ and $G'$ are exactly the same. Let $H$ be a 2-subdivision of $G$. Since $G'$ is an irreflexive digraph, by Lemma~\ref{lem:kernel_npc} we have that $G'$ has a kernel if and only if $H$ has a kernel. Also, we have by Lemma~\ref{lem:subdiv_point} that $H$ is a point-point digraph. Therefore we can conclude that $G$ has a kernel if and only if the point-point digraph $H$ has a kernel. Thus a polynomial time algorithm that solves the {\sc Kernel} problem in point-point digraphs can be used to solve the {\sc Kernel} problem in general digraphs in polynomial time. This proves the theorem.
\end{proof}

Note that {\sc Kernel} is known to be NP-complete even on planar digraphs having degree at most 3 and in- and out-degrees at most 2~\cite{fraenkel1981planar}. The above reduction transforms the input digraph in such a way that every newly introduced vertex has in- and out-degree exactly 1 and the in- and out-degrees of the original vertices remain the same. Moreover, if the input digraph is planar, the digraph produced by the reduction is also planar. Thus we can conclude that the problem {\sc Kernel} remains NP-complete even for planar point-point digraphs having degree at most 3 and in- and out-degrees at most 2.
\bigskip

An \emph{L-reduction} as defined below is an approximation-preserving reduction for optimization problems.
\begin{definition}[\cite{papadimitriou1991optimization}]\label{def:lreduction}
	Let $A$ and $B$ be two optimization problems with cost functions $c_A$ and $c_B$ respectively. Let $f$ be a polynomially computable function that maps the instances of problem $A$ to the instances of problem $B$. Then $f$ is said to be an L-reduction form $A$ to $B$ if there exist a polynomially computable function $g$ and constants $\alpha,\beta \in (0,\infty)$ such that the following conditions hold:
	\begin{myenumerate}
		\item If $y'$ is a solution to $f(x)$ then $g(y')$ is a solution to $x$, where $x$ is an instance of the problem $A$.
		\item $OPT_B(f(x))\leq \alpha OPT_A(x)$, where $OPT_B(f(x))$ and $OPT_A(x)$ denote the optimum value of respective instances for the problems $B$ and $A$ respectively.
		\item $|OPT_A(x)-c_A(g(y'))|\leq \beta |OPT_B(f(x)-c_B(y')|$.
	\end{myenumerate}
\end{definition}
In order to prove that a problem $P$ is APX-hard, it is enough to show that the problem $P$ has an L-reduction from an APX-hard problem.

\begin{theorem}\label{thm:APXkernel}
For $k\geq 1$, the problems {\sc Min-Kernel} and {\sc Max-Kernel} are APX-hard for $2k$-subdivisions of irreflexive symmetric digraphs having in- and out-degree at most 3. Consequently, the problems {\sc Min-Kernel} and {\sc Max-Kernel} are APX-hard for point-point digraphs having in- and out-degree at most 3.
\end{theorem}
\begin{proof}
	By Corollary~\ref{cor:APX}, we have that the problems {\sc Min-Kernel} and {\sc Max-Kernel} are APX-complete for irreflexive symmetric digraphs having in- and out-degree at most 3. Here we give an L-reduction from the {\sc Min-Kernel} and {\sc Max-Kernel} problems for irreflexive symmetric digraphs having in- and out-degree at most 3 to the {\sc Min-Kernel} and {\sc Max-Kernel} problems for $2k$-subdivisions of irreflexive symmetric digraphs having in- and out-degree at most 3. Let $G$ be an irreflexive symmetric digraph of in- and out-degree at most 3, where $|V(G)|=n$ and $|E(G)|=m$. For $k\geq 1$, let $H$ be the $2k$-subdivision of $G$. Clearly, $H$ can be constructed in polynomial time. And let $K(G)$ (resp. $K'(G)$) and $K(H)$ (resp. $K'(H)$) denote a minimum (resp. maximum) sized kernel in $G$ and $H$ respectively. Since $G$ is a digraph of in- and out-degree at most 3, we have that $m\leq 3n$.
	Note that every absorbing set of $G$ has size at least $\frac{n}{4}$, since each vertex has at most 3 in-neighbours. As a minimum (resp. maximum) kernel of $G$ is an absorbing set of $G$, we have $|K(G)|=q\geq \frac{n}{4}$ (resp. $|K'(G)|=q'\geq \frac{n}{4}$). By Lemma~\ref{lem:kernel_npc}, we have that  $|K(H)|= q+km$ (resp. $K'(H)=q'+km$). Therefore, $\frac{|K(H)|}{|K(G)|}\leq 1+12k$ (resp. $\frac{|K(H')|}{|K(G')|}\leq 1+12k$). We can now choose $\alpha=1+12k$ and $\beta=1$ so that our reduction satisfies the requirements of Definition~\ref{def:lreduction} (Lemma~\ref{lem:kernel_npc} guarantees that condition (3) of Definition~\ref{def:lreduction} holds, and also that the function $g$ in the definition is polynomial time computable). Thus our reduction is an L-reduction, which implies that the problems {\sc Min-Kernel} and {\sc Max-Kernel} are APX-hard for $2k$-subdivisions of irreflexive symmetric digraphs having in- and out-degree at most 3. Now by Lemma~\ref{lem:subdiv_point}, we have that the $2k$-subdivision of any irreflexive digraph $G$ is a point-point digraph. Therefore, now we can conclude that the problems {\sc Min-Kernel} and {\sc Max-Kernel} are APX-hard for point-point digraphs.
\end{proof}

\subsection{Minimum absorbing set}
\begin{lemma}\label{lem:abs_npc}
	Let $G$ be an irreflexive digraph and let $k\geq 1$. Then $G$ has an absorbing set of size at most $q$ if and only if the $2k$-subdivision of $G$ has an absorbing set of size at most $q+km$. Further, given an absorbing set of size at most $q+km$ in the $2k$-subdivision of $G$, we can construct in polynomial time an absorbing set of size at most $q$ in $G$.
\end{lemma}
\begin{proof}
	Let $H$ be the $2k$-subdivision of $G$ and let $\bigcup_{(i,j)\in E(G)} \{u_{ij}^1,u_{ij}^2,\ldots,u_{ij}^{2k}\}$ be the vertices in $V(H)\setminus V(G)$ as defined in Section~\ref{sec:subdivision}. 
	
	Suppose that $G$ has an absorbing set $A\subseteq V(G)$ such that $|A|\leq q$. We define the set $A'\subseteq V(H)$ as $A'=A\cup\bigcup_{(i,j)\in E(G)} A(i,j)$, where
	$$A(i,j)=\begin{cases}
	\{u_{ij}^{2l}: l\in \{1,2,\ldots,k\}\}, & \text{ if } j\notin A\\
	\{u_{ij}^{2l-1}: l\in \{1,2,\ldots,k\}\}, & \text{ if } j\in A
	\end{cases}$$  
	We claim that $A'$ is an absorbing set in $H$ of size at most $q+km$. Consider any $(i,j)\in E(G)$. It is clear from the definition of $A'$ that for each $t\in\{1,2,\ldots,2k-1\}$, either the vertex $u_{ij}^t$ or $u_{ij}^{t+1}$ is in $A'$. Further, we also have that either the vertex $u_{ij}^{2k}$ or $j$ is in $A'$. Thus for every vertex $x\in V(H)\setminus V(G)$, either $x$ or one of its out-neighbours is in $A'$. Now consider a vertex $i$ in $H$ such that $i\in V(G)$. If $i\in A$, then $i\in A'$. On the other hand if $i\notin A$, then since $A$ is an absorbing set in $G$, there exists an out-neighbour $j$ of $i$ such that $j\in A$, in which case we have $u_{ij}^1\in A'$. Thus in any case, either $i$ or an out-neighbour of $i$ is in $A'$. This shows that $A'$ is an absorbing set in $H$. As $A'$ is obtained from $A$ by adding exactly $k$ new vertices corresponding to each of the $m$ edges in $G$, we also have that $|A'|\leq q+km$. This proves our claim.
	
	For any set $S\subseteq V(H)$ and $(i,j)\in E(G)$, we define $S_{ij}=S\cap \{u_{ij}^1,u_{ij}^2,\ldots,u_{ij}^{2k-1},u_{ij}^{2k}\}$.
	Now suppose that $H$ has an absorbing set $A'$ of size at most $q+km$.
	Let $F=\{(i,j)\in E(G)\colon |A'_{ij}|>k\}$. Now define the set $A''=(A'\setminus\bigcup_{(i,j)\in F} A'_{ij}) \cup \bigcup_{(i,j)\in F}(\{u_{ij}^{2l-1}: l\in \{1,2,\ldots,k\}\}\cup\{j\})$. Clearly, $A''$ is also an absorbing set in $H$, $|A''|\leq|A'|\leq q+km$. Since $A''$ is an absorbing set in $H$, for $(i,j)\in E(G)$ and each $t\in \{1,2,\ldots,2k-1\}$, either $u_{ij}^t\in A''$ or $u_{ij}^{t+1}\in A''$. This implies that $|A''_{ij}|\geq k$. From the construction of $A''$, it is clear that for each $(i,j)\in E(G)$, $|A''_{ij}|\leq k$. Therefore, we can conclude that $|A''_{ij}|= k$ for each $(i,j)\in E(G)$. It then follows that for each $t\in\{1,2,\ldots,2k-1\}$, exactly one of $u_{ij}^t,u_{ij}^{t+1}$ is in $A''$. We now claim that $A=A''\cap V(G)$ is an absorbing set in $G$. Let $i\in V(G)$. Suppose that $i\notin A$, which means that $i\notin A''$. Since $A''$ is an absorbing set in $H$, we have that there exists a vertex $j\in N^+_{G}(i)$ such that $u_{ij}^1\in A''$. By our earlier observation that exactly one of $u_{ij}^t,u_{ij}^{t+1}\in A''$ for each $t\in\{1,2,\ldots,2k-1\}$, we now have that $u_{ij}^{2k}\notin A''$. This would imply that $j\in A''$. Therefore we can conclude that for any vertex $i\in V(G)$, either $i\in A$ or one of its out-neighbors is in $A$. This implies that $A$ is an absorbing set in $G$. Since $|A''_{ij}|= k$ for each $(i,j)\in E(G)$ and $|E(G)|=m$, we now have that $|A|=|A''|-km\leq q$. It is also easy to see that given the absorbing set $A'$ of $H$, we can construct $A''$ and then $A''\cap V(G)$ in polynomial time. This proves the lemma.
\end{proof}

\begin{theorem}
For $k\geq 1$, the problem {\sc Absorbing-Set} is APX-hard for $2k$-subdivisions of irreflexive symmetric digraphs having in- and out-degree at most 3. Consequently, the problem {\sc Absorbing-Set} is APX-hard for point-point digraphs having in- and out-degree at most 3.
\end{theorem}
\begin{proof}
	This can be proved in a way similar to that of Theorem~\ref{thm:APXkernel}.
	By Corollary~\ref{cor:APX}, we have that the {\sc Absorbing-Set} problem is APX-complete for irreflexive symmetric digraphs having in- and out-degree at most 3. We give an L-reduction from the {\sc Absorbing-Set} problem for irreflexive symmetric digraphs having in- and out-degree at most 3 to the {\sc Absorbing-Set} problem for $2k$-subdivisions of irreflexive symmetric digraphs having in and out-degree at most 3. Let $G$ be an irreflexive symmetric digraph of in- and out-degree at most 3, where $|V(G)|=n$ and $|E(G)|=m$. For $k\geq 1$, let $H$ be the $2k$-subdivision of $G$. Clearly, $H$ can be constructed in polynomial time. And let $A(G)$ and $A(H)$ denote a minimum sized absorbing set in $G$ and $H$ respectively. Since $G$ is a digraph of in- and out-degree at most 3, as noted in the proof of Theorem~\ref{thm:APXkernel}, we have that $m\leq 3n$ and $|A(G)|\geq \frac{n}{4}$. By Lemma~\ref{lem:abs_npc}, we have that  $|A(H)|\leq |A(G)|+km$. Therefore as $\frac{|A(H)|}{|A(G)|}\leq 1+12k$, we can now choose $\alpha=1+12k$ and $\beta=1$ so that our reduction satisfies the requirements of Definition~\ref{def:lreduction} (Lemma~\ref{lem:abs_npc} guarantees that condition (3) of Definition~\ref{def:lreduction} holds, and also that the function $g$ in the definition is polynomial time computable). Thus our reduction is an L-reduction, which implies that {\sc Absorbing-Set} is APX-hard for $2k$-subdivisions of irreflexive symmetric digraphs having in- and out-degree at most 3. Since the $2k$-subdivision of any irreflexive digraph $G$ is a point-point digraph by Lemma~\ref{lem:subdiv_point}, we can now conclude that the problem {\sc Absorbing-Set} is APX-hard for point-point digraphs.
\end{proof}
\section{Comparability relations between classes}

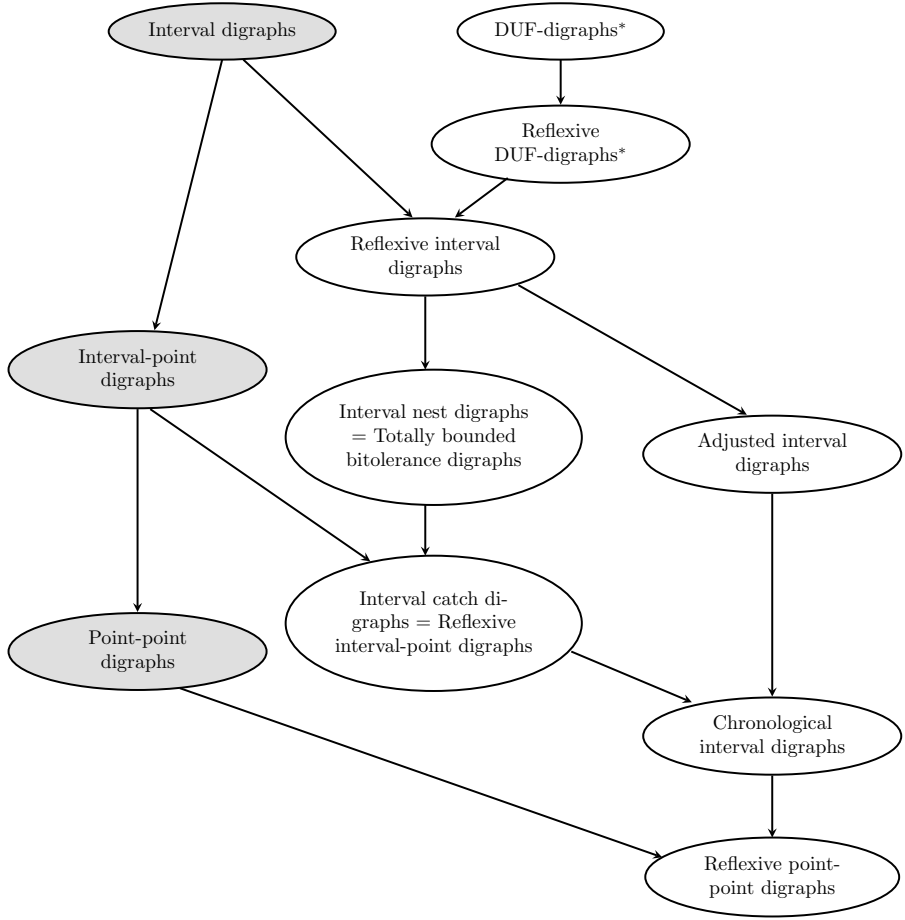
\begin{figure}
	\centering
	\scalebox{.75}{
		\begin{tikzpicture}[xscale=0.75]
\node[draw,
ellipse, line width=1,
minimum width=2cm,minimum height=1cm,fill=gray!25] at (-3,1){Interval digraphs};
\node[draw,
ellipse, line width=1, 
minimum width=2cm,minimum height=1cm] at (5,1){DUF-digraphs$^*$};
\node[draw,
ellipse, line width=1, 
minimum width=2cm,minimum height=1cm,text width=3cm,align=center] at (5,-1){Reflexive\\DUF-digraphs$^*$};
\node[draw,
ellipse, line width=1, 
minimum width=2cm,minimum height=1cm,text width=3cm,align=center] at (1.8,-3){Reflexive interval \\ digraphs};
\node[draw,
ellipse, line width=1,
minimum width=2cm,minimum height=1cm,text width=3cm,align=center,fill=gray!25] at (-5,-5){Interval-point \\ digraphs};
\node[draw,
ellipse, line width=1,
minimum width=2cm,minimum height=1cm,text width=3cm,align=center,fill=gray!25] at (-5,-10){Point-point \\ digraphs};
\draw [line width=1]
(2,-6.2) ellipse (3.5 and 1.2); 
\node[draw,
ellipse, line width=1,
minimum width=2cm,minimum height=1cm,text width=3cm,align=center] at (10,-6.5){Adjusted interval \\ digraphs};
\node[draw,
ellipse, line width=1,
minimum width=2cm,minimum height=1cm,text width=3cm,align=center] at (10,-11.5){Chronological \\ interval digraphs};
\draw[-stealth, line width=1] (-2.5,0.5) -- (1.5,-2.3);
\draw[-stealth, line width=1] (-3,0.5) -- (-4.6,-4.3);
\draw[-stealth, line width=1] (-5,-5.7) -- (-5,-9.3);
\draw[-stealth, line width=1] (-4.7,-5.7) -- (0.5,-8.4);
\draw[-stealth, line width=1] (1.8,-7.4) -- (1.8,-8.3);
\draw[-stealth, line width=1] (-4,-10.65) -- (7.4,-13.65);
\draw[-stealth, line width=1] (1.8,-3.7) -- (1.8,-5);
\draw[-stealth, line width=1] (5.25,-10)-- (8.1,-10.9);
\draw[-stealth, line width=1] (10,-12.2)-- (10,-13.3);
\draw[-stealth, line width=1] (5,.5)-- (5,-.3);
\draw[-stealth, line width=1] (3.75,-1.6)-- (2.5,-2.3);
\draw[-stealth, line width=1] (4,-3.5) -- (9.4,-5.8);
\draw[-stealth, line width=1] (10,-7.2)-- (10,-10.8);
\node[text width=4cm,align=center] at (2,-6.2) {Interval nest digraphs = Totally bounded bitolerance digraphs};
\node[text width=4cm,align=center] at (2,-9.5) {Interval catch digraphs = Reflexive interval-point digraphs};
\draw [line width=1]
(2,-9.5) ellipse (3.5 and 1.2); 
\node[text width=4cm,align=center] at (10,-14) {Reflexive point-point digraphs};
\draw [line width=1]
(10,-14) ellipse (3 and .7); 
\end{tikzpicture}
	}
	\caption{Inclusion relations between graph classes. In the diagram, there is an arrow from $\mathcal{A}$ to $\mathcal{B}$ if and only if the class $\mathcal{B}$ is contained in the class $\mathcal{A}$. Moreover, each inclusion is strict. The problems studied are efficiently solvable in the classes shown in white, while they are NP-hard and/or APX-hard in the classes shown in gray ($^*$~the complexity of the {\sc Absorbing-Set} problem on DUF-digraphs and reflexive DUF-digraphs remain open).}
	\label{hierarchy}
\end{figure}

Figure~\ref{hierarchy} shows the inclusion relations between the classes of digraphs that were studied in this paper. Note that the class of interval digraphs and the class of DUF-digraphs are incomparable to each other. This can be shown as follows: a directed triangle with edges $(a,b),(b,c),(c,a)$ is a point-point digraph, but it is easy to see that there is no DUF-ordering for this digraph. Thus, the class of point-point digraphs is not contained in the class of DUF-digraphs. On the other hand, consider a symmetric triangle $G$ with edges $(a,b),(b,a),(b,c),(c,b),(c,a),(a,c)$. Then any permutation of the vertices in $G$ is a DUF-ordering of $G$. Note that the splitting bigraph $B_G$ of $G$ is an induced cycle of length 6. If $G$ is an interval digraph, then $B_G$ is an interval bigraph, which contradicts M\"uller's observation~\cite{muller1997recognizing} that interval bigraphs are chordal bipartite graphs (bipartite graphs that do not contain any induced cycle $C_k$, for $k\geq 6$). Thus $G$ is not an interval digraph, implying that the class of DUF-digraphs is not contained in the class of interval digraphs. Further note that, even the class of reflexive DUF-digraphs is not contained in the class of interval digraphs, as otherwise every reflexive DUF-digraph should have been a reflexive interval digraph, which is not true: by Theorem~\ref{thm:K_33free}, the underlying undirected graph of a reflexive interval digraph cannot contain $K_{3,3}$ as an induced subgraph, but orienting every edge of a $K_{3,3}$ from one partite set to the other and adding a self-loop at each vertex gives a reflexive DUF-digraph (any ordering of the vertices in which the vertices in one partite set all come before every vertex in the other partite set is a DUF-ordering of this digraph). Clearly, there are DUF-digraphs that are not reflexive, implying that the class of reflexive DUF-digraphs forms a strict subclass of DUF-digraphs.

Now in~\cite{das1989interval}, the authors give an example of a digraph which is not an interval point digraph as follows: The digraph has vertex set $\{v_1,v_2,v_3,v_4\}$ and edge set $\{(v_2,v_2),(v_3,v_3),(v_4,v_4),(v_2,v_1),(v_3,v_1),(v_4,v_1)\}$. They observed that this digraph is not an interval point digraph. We slightly modify the above example by adding a loop at $v_1$ and call the resulting reflexive digraph as $G$. It is then easy to verify that the modified digraph $G$ is not an interval nest digraph. (Note that in any interval nest representation of $G$, there exists $x\in\{v_2,v_3,v_4\}$ such that $S_x\subseteq S_{v_1}\cup\bigcup_{a\in\{v_2,v_3,v_4\}\setminus \{x\}} S_a$. As $T_x \subseteq S_x$, this implies that either $(v_1,x)\in E(G)$ or there exists an $a\in \{v_2,v_3,v_4\}\setminus \{x\}$ such that $(a,x)\in E(G)$, which is a contradiction to the definition of $G$.) But consider the ordering $<\colon (v_1,v_2,v_3,v_4)$ of the vertices in $G$. It has the property that, for $i<j<k$, if $(v_i,v_k)\in E(G)$ then $(v_i,v_j)\in E(G)$, and  if $(v_k,v_i)\in E(G)$ then $(v_j,v_i)\in E(G)$. In fact, the class of reflexive digraphs that has a vertex ordering satisfying the above property is known to be the class of adjusted interval digraphs~\cite{feder2012adjusted}. This shows that $G$ as defined above is an adjusted interval digraph. Since $G$ is not an interval nest digraph, we can conclude that the class of adjusted interval digraphs (and therefore, the class of reflexive interval digraphs) is not contained in the class of interval nest digraphs (and therefore, not contained in the class of  interval catch digraphs). Since interval catch digraphs are exactly reflexive interval-point digraphs, this also means that the class of adjusted interval digraphs (and therefore, the class of reflexive interval digraphs) is not contained in the class of interval-point digraphs.

Now consider the digraph $G$ with $V(G)=\{a,b,c,d\}$ and edges $(a,b),(a,d),(c,b),(c,d)$ in addition to loops at each vertex. It is easy to construct an interval catch representation of $G$. But note that the underlying undirected graph of $G$ is an induced $C_4$. This implies that $G$ is not an adjusted interval digraph, as otherwise it contradicts the fact that the underlying undirected graphs of adjusted interval digraphs are interval graphs~\cite{feder2012adjusted}. This proves that the class of interval catch digraphs (and therefore, the class of reflexive interval digraphs) is not contained in the class of adjusted interval digraphs.

Now let $G$ be a digraph with $V(G)=\{a,b,c,d\}$ and edges $(a,b),(c,b),(b,d),(d,b)$ in addition to loops at each vertex. The digraph $G$ is not an interval catch digraph, as in any interval catch representation of $G$, the point $T_b$ contained in each of the intervals $S_a,S_b$ and $S_c$. Thus the intervals $S_a,S_b,S_c$ intersect pairwise, which implies that one of the intervals $S_a,S_b,S_c$ is contained in the union of the other two. We have that $S_a$ is not contained in $S_b\cup S_c$, since otherwise the fact that $T_a\in S_a$ implies that either $(b,a)$ or $(c,a)$ is an edge in $G$, which is a contradiction. For the same reason, we also have that $S_c$ is not contained in $S_a\cup S_b$. We can therefore conclude that $S_b\subseteq S_a\cup S_c$. But as $(b,d)\in E(G)$, we have that $T_d\in S_b$, which implies that either $(a,d)$ or $(c,d)$ is an edge in $G$---a contradiction. Thus $G$ is not an interval catch digraph. On the other hand, it can be seen that $G$ is an interval nest digraph (one possible interval nest representation of $G$ is as follows: $S_a=[1,2]$, $S_b=T_b=[2,4]$, $S_c=[4,5]$, $S_d=T_d=[3,3]$, $T_a=[1,1]$ and $T_c=[5,5]$). Thus the class of interval nest digraphs is not contained in the class of interval catch digraphs (and therefore not contained in the class of interval-point digraphs, as interval-point digraphs are exactly the reflexive interval catch digraphs).

Consider a digraph $G$ with $V(G)=\{a,b,c,d\}$ and edges $(a,b),(a,c),(b,c),(c,b),(c,d)$ in addition to loops at each vertex. It is easy to construct a chronological interval representation for $G$. But as $(a,b),(c,b),(c,d)\in E(G)$ and $(a,d)\notin E(G)$, we have that $a,b,c,d$ is an anti-directed walk of length 3. Therefore by Theorem~\ref{thm:point-point}, we have that $G$ is not a point-point digraph. Thus we have that the class of chronological interval digraphs is not contained in the class of point-point digraphs.

Finally, note that the class of reflexive point-point digraphs is nothing but the class of all digraphs that can be obtained by the disjoint union of complete digraphs (a complete digraph is the digraph in which there is a directed edge from every vertex to every other vertex and itself). The class of interval nest digraphs coincides with the class of totally bounded bitolerance digraphs which was introduced by Bogart and Trenk~\cite{bogart2000bounded}.

The above observations and the definitions of the corresponding classes explain the comparability and the incomparability relations for the classes of digraphs shown in Figure~\ref{hierarchy}.

\section{Conclusion}\label{sec:conclusion}
After work on this paper had been completed, we have been made aware of a recent manuscript of Jaffke, Kwon and Telle~\cite{jaffke2021classes}, in which unified polynomial time algorithms have been obtained for the problems considered in this paper for some classes reflexive intersection digraphs including reflexive interval digraphs. Their algorithms are more general in nature, and consequently have much higher time complexity, while our targeted algorithms are much more efficient; for example our algorithm finds a minimum dominating (or absorbing) set in a reflexive interval digraph in time $O(m+n)$, while the general algorithm of~\cite{jaffke2021classes} has complexity $O(n^8)$.
As noted above, totally bounded bitolerance digraphs are a subclass of reflexive interval digraphs, and therefore all the results that we obtain for reflexive interval digraphs hold also for this class of digraphs. 

M\"uller~\cite{muller1997recognizing} showed the close connection between interval digraphs and interval bigraphs and used this to construct the only known polynomial time recognition algorithm for both these classes. Since this algorithm takes $O(nm^6(n+m)\log n)$ time, the problem of finding a forbidden structure characterization for either of these classes or a faster recognition algorithm are long standing open questions in this field. But many of the subclasses of interval digraphs, like adjusted interval digraphs~\cite{takaoka2021recognition}, chronological interval digraphs~\cite{das2013recognition}, interval catch digraphs~\cite{prisner1989characterization}, and interval point digraphs~\cite{prisner1994algorithms} have simpler and much more efficient recognition algorithms. It is quite possible that simpler and efficient algorithms for recognition exist also for reflexive interval digraphs. As for the case of interval nest digraphs, no polynomial time recognition algorithm is known. The complexities of the recognition problem and {\sc Absorbing-Set} problem for DUF-digraphs also remain as open problems.
\bigskip

\noindent\textbf{Acknowledgements.} Pavol Hell would like to gratefully acknowledge support from an NSERC Discovery Grant. Part of the work done by Pavol Hell was also supported by the Smt Rukmini Gopalakrishnachar Chair Professorship at the Indian Institute of Science (November--December, 2019).
\bibliographystyle{plain}
\bibliography{reference} 
\end{document}